\documentclass[12pt]{article}%
\usepackage{latexsym,amsmath,amssymb,amsthm,amsfonts,graphicx,natbib}
\usepackage{epsfig}
\usepackage[mathlines]{lineno}
\usepackage{amsmath}
\usepackage{amsfonts}
\usepackage{amssymb}
\usepackage{graphicx}%
\usepackage{setspace}
\providecommand{\U}[1]{\protect\rule{.1in}{.1in}}
\include{def}
\def\dispmuskip{\thinmuskip= 3mu plus 0mu minus 2mu \medmuskip=  4mu plus 2mu minus 2mu \thickmuskip=5mu plus 5mu minus 2mu}
\def\textmuskip{\thinmuskip= 0mu                    \medmuskip=  1mu plus 1mu minus 1mu \thickmuskip=2mu plus 3mu minus 1mu}
\def\beq{\dispmuskip\begin{equation}}
\def\eeq{\end{equation}\textmuskip}
\def\beqn{\dispmuskip\begin{displaymath}}
\def\eeqn{\end{displaymath}\textmuskip}
\def\bea{\dispmuskip\begin{eqnarray}}
\def\eea{\end{eqnarray}\textmuskip}
\def\bqan{\dispmuskip\begin{eqnarray*}}
\def\eqan{\end{eqnarray*}\textmuskip}
\def\paradot#1{\vspace{1.3ex plus 0.7ex minus 0.5ex}\noindent{\bf\boldmath{#1.}}}
\newtheorem{theorem}{Theorem}

\newtheorem{corollary}{Corollary}

\newtheorem{definition}{Definition}

\newtheorem{lemma}{Lemma}

\newcommand{\lpds}{log predictive density score}

\newcommand{\KL}{Kullback-Liebler}
\newcommand{\MH}{Metropolis-Hastings}

\newcommand{\wh}{\widehat}
\newcommand{\wt}{\widetilde}
\def\v{\boldsymbol}

\def\E{{\mathbb{E}}}
\def\P{{\mathbb{P}}}

\def\I{1\!\!1}
\def\v{\boldsymbol}

\def\t{\theta}

\def\H{{\cal H}}

\def\LPDS{\text{\rm LPDS}}
\def\KL{\text{\rm KL}}

\def\tr{\text{\rm trace}}

\def\argmax{\text{\rm argmax}}

\def\diag{\text{\rm diag}}
\newcommand{\pihat}{\widehat \pi}
\newcommand{\vb}{\text{vb}}

\newcommand{\mutilde} {\widetilde \mu}
\newcommand{\nutilde} {\widetilde \nu}
\newcommand{\Sigmatilde} {\widetilde \Sigma}
\newcommand{\khat}{\widehat k}

\newcommand{\calHn}{{\cal H}^{n-1}}
\newcommand{\psihat}{\widehat \psi}
\newcommand{\IG}{\text{\rm IG}}
\newcommand{\hmax}{h_{\rm max}}

\begin{document}
\def\spacingset#1{\renewcommand{\baselinestretch}%
{#1}\small\normalsize} \spacingset{1}

%\linenumbers
\title{Adaptive Metropolis-Hastings Sampling using Reversible Dependent Mixture
Proposals}
\author{Minh-Ngoc Tran\\\textit{\ }{\small Australian School of Business}\\{\small University of New South Wales, Australia}\\{\small \texttt{minh-ngoc.tran@unsw.edu.au}}\\
\and Michael K. Pitt\\{\small Department of Economics}\\{\small University of Warwick, United Kingdom}\\{\small \texttt{M.Pitt@warwick.ac.uk}}\\
\and Robert Kohn\\{\small Australian School of Business}\\{\small University of New South Wales, Australia}\\{\small \texttt{r.kohn@unsw.edu.au}} }
\date{}
\maketitle

\begin{abstract}
This article develops a general-purpose adaptive sampler that approximates the
target density by a mixture of multivariate $t$ densities. The adaptive
sampler is based on reversible proposal distributions each of which has the
mixture of multivariate $t$ densities as its invariant density. The reversible
proposals consist of a combination of independent and correlated steps that
allow the sampler to traverse the parameter space efficiently as well as
allowing the sampler to keep moving and locally exploring the parameter space.
We employ a two-chain approach, in which a trial chain is used to adapt the
proposal densities used in the main chain. Convergence of the main chain
and a strong law of large numbers are proved under reasonable conditions, and
without imposing a Diminishing Adaptation condition. The mixtures of
multivariate $t$ densities are fitted by an efficient Variational
Approximation algorithm in which the number of components is determined
automatically. The performance of the sampler is evaluated using simulated and
real examples. Our autocorrelated framework is quite general and can handle mixtures
other than  multivariate $t$.

\vspace{1.3ex plus 0.7ex minus 0.5ex}\noindent\textbf{\boldmath{Keywords.}}
 Ergodic convergence; Markov Chain Monte Carlo; Metropolis-within Gibbs composite sampling;
Multivariate
$t$ mixtures; Simulated annealing; Variational Approximation.

\end{abstract}

%\pagewiselinenumbers

%======================================================================%
\newpage

%======================================================================%

\section{Introduction}

\label{S: introd}
\onehalfspacing
Suppose that we wish to sample from a target distribution using a
Metropolis-Hastings sampling method. For a traditional Metropolis-Hastings{}
algorithm, where the proposal distribution is fixed in advance, it is well
known that the success of the sampling method depends heavily on how the
proposal distribution is selected. It is challenging to develop non-adaptive proposals
in several types of problems. One example is when the target density is highly
non-standard and/or multimodal. A second example is when the parameters are structural and
deeply embedded in the likelihood so that it is difficult to differentiate the likelihood
with respect to the likelihood; see for example \cite{schmidl:2013} who consider
dynamic model with a regression function  that is obtained as a solution to a differential equation.
In such cases adaptive sampling, which sequentially updates the
proposal distribution based on the previous iterates, has been shown useful.
See, e.g., \cite{Haario:1999,Haario:2001,Roberts:2007,Roberts:2009,Holden:2009} and
\cite{Giordani:2010}.

The chain generated from an adaptive sampler is no longer Markovian, and the
convergence results obtained for traditional Markov chain Monte Carlo (MCMC)
sampling no longer apply. \cite{Andrieu:2008} warn that care must be taken in
designing an adaptive sampler, as otherwise it may not converge to the correct
distribution. They demonstrate this by constructing an adaptive chain that
does not converge to the target. However, some theoretical results on the
convergence of adaptive samplers are now available. A number of papers prove
convergence by assuming that the adaptation
eventually becomes negligible, which they call the Diminishing Adaptation
condition. See, for example,
\cite{Haario:1999,Haario:2001,Andrieu:2006,Roberts:2007} and
\cite{Giordani:2010}. This condition is relatively easy to check if the
adaptation is based on moment estimates as, for example, in an adaptive random
walk, as in \cite{Haario:1999,Haario:2001} and \cite{Roberts:2007} where the
rate at which the adaptation diminishes is governed naturally by the sample
size. It is more difficult to determine an optimal rate of adaptation when the
adaptation is based on non-quadratic optimization as in \cite{Giordani:2010}
and in the Variational Approximation approach in our article.

Our article constructs a general-purpose adaptive sampler that we call the
Adaptive Correlated Metropolis-Hastings (ACMH) sampler. The sampler is
described in Section~\ref{Sec:adaptive}. The first contribution of the article
is to propose a two-chain approach to construct the proposal densities, with
the iterates of the first chain used to construct the proposal densities used
in the second (main) chain. The ACMH sampler approximates the target density
by a sequence of mixtures of multivariate $t$ densities. The heavy tails of
$t$ and mixture of $t$ distributions is a desirable property that a proposal
distribution should have. Each mixture of $t$ distribution is fitted by a
Variational Approximation algorithm which automatically selects the number of
components. Variational Approximation is now well known as a computationally
efficient method for estimating complex density functions; see, e.g.,
\cite{McGrory:2007} and \cite{Giordani:2012}. An attractive property of
Variational Approximation that makes it suitable for constructing proposal
distributions is that it can locate the modes quickly and efficiently.

The second, and main,  contribution of the article is to introduce in
Section~\ref{Sec:constructing} a method to construct reversible proposal
densities, each of which has the mixture of $t$ approximation as its invariant
density. The proposal densities consist of both autocorrelated and independent
Metropolis-Hastings steps. Independent steps allow the sampler to traverse the
parameter space efficiently, while correlated steps allow the sampler to keep
moving, (i.e., avoid getting stuck), while also exploring the parameter
space locally. If the approximating $t$ mixture is close to the target, then the
reversible proposals introduced in our article will allow the sampler to move
easily and overcome the low acceptance rates often encountered by purely
independent proposals. Note that the reversible correlated proposals we
introduce are quite general and it is only necessary to be able to generate from
them but it is unnecessary to be able to evaluate them. This is an important
property for the correlated mixtures of $t$ that we use.

The third contribution of the paper is to show in Section~\ref{Sec:theory}
that the ACMH sampler converges uniformly to the target density and to obtain
a strong law of large numbers under reasonable conditions,
without requiring that Diminishing Adaptation holds. As pointed out above, it
is difficult to impose Diminishing Adaptation in a natural way for general
proposals such as those in our paper.

Adaptive sampling algorithms can be categorized into two groups: exploitative
and exploratory algorithms \citep{Schmidler:2011}. Exploitative algorithms
attempt to improve on features of the target distribution that have already
been seen by the sampler, i.e. based on the past iterations to improve on what
have been discovered by the past iterations. The adaptive samplers of
\cite{Haario:2001} and \cite{Giordani:2010} belong to this group. The second
group encourages exploring the whole support of the target, including
tempering \citep{Geyer:1995}, simulated annealing \citep{Neal:2001} and the
Wang-Landau algorithm \citep{Wang:2001a,Wang:2001b}. It is therefore useful to
develop a general-purpose adaptive sampler that can be both exploratory and
exploitative. An important feature of the ACMH sampler is the use, in
Section~\ref{Sec:initialization},  of an exploratory stage to initialize the
adaptive chain. In particular, we describe in this paper how to use
simulated annealing \citep{Neal:2001}
to initialize the chain. \cite{Giordani:2010} suggest initializing the chain
using either random walk steps or by using a Laplace approximation, neither of
which work well for targets that are multimodal and/or have a non-standard
support. Initializing by an exploratory algorithm helps the sampler initially
explore efficiently the features of the target, and these features will be
improved in the subsequent exploitative stage.
Section \ref{Sec:simulations} shows that such a combination makes the ACMH
sampler work well for challenging targets where many other samplers may fail.

A second feature of the ACMH sampler is that it uses a small proportion of
adaptive random walk steps in order to explore tail regions around local modes
more effectively. A third feature of the ACMH sampler is that it
uses Metropolis-within-Gibbs component-wise sampling to make the sampler move more efficiently in high dimensions,
where it is often difficult to efficiently move the whole state vector as a single component
 because of the large differences in the values of the target and the proposal  at
the current and proposed states. See \cite{Johnson:2011} and its references
for some convergence results on such composite MCMC sampling.

Section \ref{Sec:simulations} presents simulation studies and Section
\ref{Sec:applications} applies the adaptive sampling scheme to estimate the
covariance matrix for a financial data set and analyze a spam email data set.

There are two  important and immediate extensions of our work, which are discussed in
Section~\ref{Sec:conclusion}. The first is to more general reversible mixture proposals. The second
is to problems where the likelihood cannot be evaluated explicitly, but can be estimated
unbiasedly.

\cite{Giordani:2010} construct a general-purpose adaptive independent
Metropolis-Hastings{} sampler that uses a mixture of normals as the proposal
distribution. Their adaptive sampler works well in many cases because it is
flexible and so helps the proposal approximate the target distribution better.
They use the $k$-means algorithm to estimate the mixtures of normals. Although
this method is fast, using independent Metropolis-Hastings{} steps only may
result in a low acceptance rate that may not explore the local features of the
state space as effectively. In addition, the automatic selection of components
used in our article works appreciably better than using BIC, as done in
\cite{Giordani:2010}.

\cite{Freitas:2001} use Variational Approximation to first estimate the target
density and then use this approximation to form a fixed proposal density
within an MCMC scheme. There are two problems with this approach, which are
discussed in Section~\ref{Sec:estimating mixt}.

\cite{Holden:2009} provide a framework for constructing adaptive samplers that
ensures ergodicity without assuming Diminishing Adaptation. Not imposing
Diminishing Adaptation is attractive because it means that the adaptation can
continue indefinitely if new features of the target are learned. However, we
believe that the \cite{Holden:2009} framework is unnecessarily limited for two
reasons. First, it does not use the information about the target obtained from
dependent steps. Second, it augments the history on which the adaptation is
based by using proposals that are rejected by the Metropolis-Hastings method;
such inclusions typically lead to suboptimal adaptive proposals in our experience.

\cite{Hoogerheide2012} also use a multivariate mixture of $t$ proposal
densities which they fit using the EM algorithm. However, they stop adapting
after a preliminary stage and do not have a principled way of choosing the
number of components. In addition, their approach is harder to use when the
likelihood cannot be computed, but can be estimated unbiasedly.
\cite{schmidl:2013} propose an adaptive approach based on a vine copula but they stop
adapting after a fixed number of adaptive steps. We note that
it is straightforward to extend the  multivariate
mixture of $t$ approach in \cite{Hoogerheide2012} and the copula approach in
\cite{schmidl:2013} to reversible proposals
and to a two chain adaptive solution as in our article.

\cite{craiu:2009} also emphasize the importance of initial exploration and
combine  exploratory and exploitative stages using
 parallel an inter-chain adaptation algorithm to initially explore the sample space.
They run many chains in parallel and let them interact in order to explore the modes,
while we use annealed sampling with an SMC sampler.
Their main contribution is the regional adaption algorithm, but they
only discuss the case with two regions/two modes.
It does  not seem straightforward to extend the algorithm to general multimodal cases
and it seems difficult to determine the number of regions/modes.

%======================================================================%

\section{The adaptive sampling framework}

\label{Sec:theory}
%======================================================================%

\subsection{Adaptive sampling algorithm}

Let $\Pi(z)$ be the target distribution with corresponding density $\pi(z)$.
We consider using a two-chain approach to adapt the proposal densities. The
idea is to run simultaneously two chains, a trial chain $X^{\prime}$ and a
main chain $X$, where the proposal densities used by the main chain $X$ are
estimated from the iterates of the trial chain $X^{\prime}$, and are not based
on the past iterates of chain $X$. We refer to the past iterates of
$X^{\prime}$ used to estimate the proposal as the history vector, and denote
by $\mathcal{H}^{n-1}$ the history vector obtained after iteration $n-1$,
which is used to compute the proposal density $q_{\mathcal{H}^{n-1}}%
(\cdot|\cdot)$ at iteration $n$. We consider the following general adaptive
sampling algorithm.

\vspace{1.3ex plus 0.7ex minus 0.5ex}\noindent\textbf{\boldmath{Two-chain
sampling algorithm.}}

\begin{enumerate}
\item Initialize the history $\mathcal{H}^{0}$, the proposal $q_{\mathcal{H}%
^{0}}(\cdot|\cdot)$ and initialize $x_{0}^{\prime}$, $x_{0}$ of the trial chain
$X^{\prime}$ and main chain $X$, respectively.

\item For $n = 1,2,...$

\begin{itemize}
\item[(a)] Update the trial chain:

\begin{itemize}

\item Generate a proposal $z^{\prime}\sim q_{\mathcal{H}^{n-1}}(z^{\prime
}|x_{n-1}^{\prime})$.

\item Compute  \thinmuskip= 3mu plus 0mu minus 2mu \medmuskip= 4mu plus 2mu
minus 2mu \thickmuskip=5mu plus 5mu minus 2mu%
\[
\alpha_{n}^{\prime}(z^{\prime},x_{n-1}^{\prime},\mathcal{H}^{n-1})=\min\left(
1,\frac{\pi(z^{\prime})q_{\mathcal{H}^{n-1}}(x_{n-1}^{\prime}|z^{\prime})}%
{\pi(x_{n-1}^{\prime})q_{\mathcal{H}^{n}}(z^{\prime}|x_{n-1}^{\prime})}\right)
.
\]
\thinmuskip= 0mu \medmuskip= 1mu plus 1mu minus 1mu \thickmuskip=2mu plus 3mu
minus 1mu

\item Accept $z^{\prime}$ with probability $\alpha_{n}^{\prime}(z^{\prime
},x_{n-1}^{\prime},\mathcal{H}^{n-1})$ and set $x_{n}^{\prime}=z^{\prime}$,
otherwise set $x_{n}^{\prime}=x_{n-1}^{\prime}$.

\item Set $\mathcal{H}^{n}=(\mathcal{H}^{n-1},x_{n-1}^{\prime})$ if
$z^{\prime}$ is accepted, otherwise set $\mathcal{H}^{n}=\mathcal{H}^{n-1}$.
\end{itemize}

\item[(b)] Update the main chain:

\begin{itemize}

\item Generate a proposal $z\sim q_{\mathcal{H}^{n-1}}(z|x_{n-1})$.

\item Compute  \thinmuskip= 3mu plus 0mu minus 2mu \medmuskip= 4mu plus 2mu
minus 2mu \thickmuskip=5mu plus 5mu minus 2mu%
\[
\alpha_{n}(z,x_{n-1})=\min\left( 1,\frac{\pi(z)q_{\mathcal{H}^{n-1}}%
(x_{n-1}|z)}{\pi(x_{n-1})q_{\mathcal{H}^{n-1}}(z|x_{n-1})}\right) .
\]
\thinmuskip= 0mu \medmuskip= 1mu plus 1mu minus 1mu \thickmuskip=2mu plus 3mu
minus 1mu

\item Set $x_{n}=z$ with probability $\alpha_{n}(z,x_{n-1})$, otherwise set
$x_{n}=x_{n-1}$.
\end{itemize}
\end{itemize}
\end{enumerate}

The ACMH sampler is based on this two-chain sampling framework and its
convergence is justified by Corollary \ref{corollary:ACMH_convergence}.

%-------------------------------------------%

\subsection{Convergence results} \label{SS: convergence results}

%-------------------------------------------%
This section presents some general convergence results for adaptive MCMC.
Suppose that $E$ is a sample space with $x,z\in E$. $\mathcal{E}$ is a
$\sigma$-field on $E$. Suppose that $q_{i}(z|x)$ is the proposal density used
at the $i$th iteration of the Metropolis-Hastings{} algorithm. In the
two-chain algorithm above, $q_{i}$ is the density $q_{\mathcal{H}^{i-1}}$
which is estimated based on the history $\mathcal{H}^{i-1}$.
Let $\{x_{n},n\geq0\}$ be the Markov chain generated by the
Metropolis-Hastings{} algorithm and ${\mathbb{P}}^{n}(x_{0},\cdot)$ the
distribution of the state $x_{n}$ with the initial state $x_{0}$. Denote by
$p_{i}(x_{i-1},dx_{i})$ the Markov transition distribution at the $i$th
iteration. We have the following convergence results whose proofs are in the Appendix.

\begin{theorem}[Ergodicity]\label{ergodic_theorem}
%-------------------------%
Suppose that
\beq\label{assumption}
p_i(x_{i-1},dx_i)\geq\beta\Pi(dx_i),\;\;\text{for all}\;\; i\geq1,
\eeq
with $0<\beta<1$.
Then
\beq\label{e:convergence}
\|\P^n(x_0,\cdot)-\Pi(\cdot)\|_\text{TV}\leq 2(1-\beta)^n\to0 \;\text{as}\;n\to\infty,
\eeq
for any initial $x_0$,
where $\|\cdot\|_\text{TV}$ denotes the total variation distance.
\end{theorem}

%-------------------------%
\begin{theorem}[Strong law of large numbers]\label{SLLN_theorem}
%-------------------------%
Suppose that $h(x)$ is a bounded function on $E$ and that \eqref{assumption} holds.
Let $S_n = \sum_{i=1}^nh(x_i)$. Then,
\beq\label{a.s.convergence}
\frac{S_n}{n}\to \E_\Pi(h)\;\;\text{almost surely}.
\eeq
\end{theorem}

\begin{corollary}\label{corollary}
Suppose that $h$ is bounded.
Theorems \ref{ergodic_theorem} and \ref{SLLN_theorem} hold for each of the following cases.

\begin{enumerate}
\item[(i)] $q_i(z|x)\geq\beta\pi(z)$ with $0<\beta<1$ for all $i\geq 1$ and $x,z\in E$.
\item[(ii)] The proposal density $q_{i}(z|x)$ is a mixture of the form
\beqn
q_{i}(z|x)= \omega q_{1,i} (z|x) + (1-\omega)q_{2,i}(z|x),\; 0<\omega<1,
\eeqn
and $q_{1,i}(z|x) \geq\beta\pi(z) $ with $0<\beta< 1$ for all $z\in E$.
\item[(iii)] Let $p_{1,i}(x_{i-1},dx_{i})$ and $p_{2,i}(x_{i-1},dx_{i})$ be
transition distributions, each has stationary distribution $\pi$.
Suppose that $p_{1,i}(x_{i-1},dx_{i})$ is based on the proposal
density $q_{1,i} (z|x)$, where $q_{1,i}(z|x) \geq\beta\pi(z)$ for all $x,z\in E$ and $0< \beta< 1$.
The transition $p_{i}(x_{i-1},dx_{i})$ at the $i$th iterate
is a mixture of the form
\beqn
p_{i}(x_{i-1},dx_{i})=\omega p_{1,i}(x_{i-1},dx_{i})+(1-\omega)p_{2,i}(x_{i-1},dx_{i}),\;\;0<\omega<1.
\eeqn
\item[(iv)] Let $p_{1,i}$ and $p_{2,i}$ be the transition distributions as in (iii).
The transition at the $i$th iterate is a composition of the form
\beqn
p_{i}(x_{i-1},dx_{i})=p_{1,i}p_{2,i}(x_{i-1},dx_{i})=\int_{z}%
p_{1,i}(x_{i-1},dz)p_{2,i}(z,dx_{i}),
\eeqn
or
\beqn
p_{i}(x_{i-1},dx_{i})=p_{2,i}p_{1,i}(x_{i-1},dx_{i})=\int_{z}%
p_{2,i}(x_{i-1},dz)p_{1,i}(z,dx_{i}),
\eeqn
\item[(v)] Let $p_{1,i}$ and $p_{2,i}$ be the transition distributions as in (iii).
The transition at the $i$th iterate is a composition of $m_{1}$ repetitions of $p_{1,i}$ and $m_{2}$
repetitions of $p_{2,i}$, i.e. $p_{i}=p_{1,i}^{m_{1}}p_{2,i}^{m_{2}}$ or
$p_i=p_{2,i}^{m_{2}}p_{1,i}^{m_{1}}$.
\end{enumerate}
\end{corollary}

%======================================================================%
\section{Reversible proposals}\label{Sec:constructing}
%======================================================================%
In the \MH{} algorithm, it is desirable to have a proposal that depends
on the current state, is reversible, and marginally has an invariant distribution of choice.
We refer to such a proposal as a reversible proposal.
The dependence between the current and proposed states helps in moving locally and helps
the chain mix more rapidly and converge.
As will be seen later, reversibility simplifies the acceptance probability
in the MH{} algorithm and makes it close to one if the marginal distribution is a good approximation to the target.
Reversibility also means that it is only necessary to be able to generate from the proposal
distribution, and it is unnecessary to be able to evaluate it. This is important in our case
because the proposal densities are mixtures of conditional $t$ densities with dependence parameters
that are integrated out.
Section \ref{subsec:theory for reversible proposals} provides the theory for reversible proposals that we use
in our article.
Section \ref{subsec:correlated} introduces a reversible multivariate $t$ density and
a reversible mixture of multivariate $t$ densities. The proofs of all results in this section are in the Appendix.

%----------------------------------------------------------------------%
\subsection{Some theory for reversible proposals}\label{subsec:theory for reversible proposals}
%----------------------------------------------------------------------%
\begin{definition}[Reversible transition density]\label{def: reversible transition density}
Suppose that $\zeta(z)$ is a density in $z$ and $T(z|x)$ is a transition
density from $x$ to $z$ such that
\[
\zeta(x)T(z|x)=\zeta(z)T(x|z)\text{ \ for any }x\text{ and }z.
\]
Then,
\[
\zeta(z)=\int\zeta(x)T(z|x)dx,
\]
and we say that $T(z|x)$ is a reversible Markov transition density with
invariant density $\zeta(z)$.
\end{definition}

The following two lemmas provide some properties of reversible transition
densities that are used in our work.

\begin{lemma}
[Properties of reversible transition densities]%
\label{lemma: properties of reversible densities} Suppose that $\zeta(z)$ is a
density in $z$. Then, in each of the cases described below, $T(z|x)$ is a
reversible Markov transition density with invariant density $\zeta(z)$.

\begin{itemize}
\item [(i)]$T(z|x)=\zeta(z)$.

\item [(ii)] Let $z = (z_{A},z_{B})$ be a partition of $z$ and define $T(z|x)=\zeta
(z_{A}|z_{B})I(z_{B}=x_{B}),$ where $I(z_{B}=x_{B})$ is an indicator variable
that is 1 if $z_{B}=x_{B}$ and is zero otherwise.

\item [(iii)]Suppose that for each parameter value $\rho$, $T(z|x;\rho)$ is a reversible
Markov transition density with invariant density $\zeta(z)$. Let $T(z|x)=\int
T(z|x;\rho)\lambda(d\rho)$, where $\lambda(d\rho)$ is a probability measure in
$\rho$.
\end{itemize}
\end{lemma}

The next lemma gives a result on a mixture of transition densities each having
its own invariant density.

\begin{lemma}[Mixture of reversible transition  densities]\label{lemma: reversible mixture}
\begin{itemize}
\item [(i)]
Suppose that for each $k=1,\dots,G$, $T_{k}(z|x)$ is a reversible Markov
transition kernel with invariant density $\zeta_{k}(z)$. Define the mixture
density $\zeta(z)$ and the mixture  $T(z|x)$ of transition densities as
\[
\zeta(z)=\sum_{k=1}^{G}\omega_{k}\zeta_{k}(z)\text{ and }T(z|x)=\sum_{k=1}%
^{G}\omega(k|x)T_{k}(z|x),
\]
where $\omega_1 + \cdots + \omega_G = 1$, $\omega_{k}\geq0$ and $\omega
(k|x)=\omega_{k}\zeta_{k}(x)/\zeta(x)$ for all $k=1,\dots,G$. Then, $T(z|x)$
is a reversible Markov transition density with invariant density $\zeta(z)$.
\item [(ii)] If the invariant densities $\zeta_k(z)$ are all the same, then $\omega(k|x) = \omega_k$  for all $k$
and $\zeta(z) = \zeta_1(z)$.
\item [(iii)]
Suppose that $T(z|x)$ is a reversible Markov transition density with invariant density $\zeta(z)$. Then
$q(z|x) = \omega \zeta(z) + (1- \omega ) T(z|x)$, $0\leq\omega\leq1$, is a reversible Markov transition density
with invariant density $\zeta (z)$.
\end{itemize}
\end{lemma}

\begin{corollary} [Mixture of conditional densities]\label{corr: mixt cond}
Let $z = (z_{A},z_{B})$ be a partition of $z$ and define $T_k(z|x)=\zeta_k
(z_{A}|z_{B})I(z_{B}=x_{B}),$ for $k=1, \dots, G$. Then, each  $T_k(z|x)$  is a reversible density
with invariant density $\zeta_k(z)$, and Lemma~\ref{lemma: reversible mixture} holds.
\end{corollary}

The next lemma shows the usefulness of using reversible Markov transition
densities as proposals in a Metropolis-Hastings scheme.

\begin{lemma}[Acceptance probability for a reversible proposal]
\label{lem: reversible accept} Consider a target density $\pi(z)$. We
propose $z$, given the state $x$ from the reversible Markov transition density
$T(z|x)$, which has invariant density $\zeta(z)$. Then, the acceptance
probability of the proposal is
\[
\alpha(z,x)=\min\left\{  1,\frac{\pi(z)\zeta(x)}{\pi(x)\zeta(z)}\right\}  .
\]
\end{lemma}

The lemma shows that the acceptance probability has the same form as for an
independent proposal from the invariant density $\zeta(z)$, even though the
proposal may depend on the previous state and on parameters that are in the
transition density $T(z|x)$ but not in $\zeta(z)$. This means the following:
(i) To compute the acceptance probability $\alpha(z,x),$ it is only necessary to be
able to simulate from $T(z|x),$ and it is unnecessary to be able to compute it.
This is useful for our work where we cannot evaluate $T(z|x)$ analytically
because it is a mixture over a parameter $\rho$.
(ii) The acceptance probability will be high if the invariant density $\zeta(z)$
is close to the target density $\pi(z)$. In fact, if $\zeta(z) = \pi(z)$, then
the acceptance probability is 1.

%----------------------------------------------------------------------%
\subsection{Constructing reversible $t$ distributions}\label{subsec:correlated}
%----------------------------------------------------------------------%
\cite{Pitt:2006} construct a
univariate Markov transition which has a univariate $t$ distribution as the invariant
distribution. We now extend this approach to construct a Markov transition density with a
multivariate $t$ density as its invariant distribution.
This reversible multivariate $t$ process is new to the literature.
We then generalize it to the case in which the invariant
distribution is a mixture of multivariate $t$ distributions.
We denote by $t_d(z;\mu,\Sigma,\nu)$ the $d$-variate $t$ density
with location vector $\mu$, scale matrix $\Sigma$ and degrees of freedom $\nu$.
\begin{lemma}[Reversible $t$ transition density]\label{lem: reversible t}
Let $\zeta(z;\psi)=t_{d}(z;\mu,\Sigma,\nu)$ and $T(z|x;\psi,\rho)=t_{d}(z; \mutilde(x)
,\Sigmatilde(x),\widetilde{\nu})$, where $\psi$ is the set of parameters
$\{\mu,\Sigma,\nu\}$, $\rho$ is a correlation coefficient,
\begin{align} \label{eq: tilde defs}
\mutilde(x) &  = (1-\rho )\mu +\rho x,\;\;\; \Sigmatilde(x)  = \frac{\nu }{%
\nu +d}(1-\rho ^{2})\left( 1+\frac{1}{\nu }(x-\mu )^\prime \Sigma ^{-1}(x-\mu
)
\right) \Sigma,\;\;\;\; \nutilde =\nu +d.
\end{align}
Then,
\begin{itemize}
\item  [(i)] For each fixed $\rho$, $T(z|x;\psi, \rho)$ is a reversible Markov
transition density with invariant density $\zeta(z;\psi)$.  \item [(ii)]  Let
$T(z|x;\psi)=\int T(z|x;\psi,\rho)\lambda(d\rho),$ where $\lambda(d\rho)$ is a
probability measure. Then, $T(z|x;\psi)$ is a reversible Markov transition
density with invariant density $\zeta(z;\psi)$.
\end{itemize}
\end{lemma}

We now follow Lemma~\ref{lemma: reversible mixture} and define a reversible transition density
that is a mixture of reversible $t$ transition densities. Suppose
$\zeta_k(z; \psi_k) = t_d(z; \mu_k, \Sigma_k, \nu_k)$ and $T_k(z|x;\psi_k,\rho_k)=t_{d}(z;\mutilde_k(x), \Sigmatilde_k(x),\nutilde_k)$,
where $\mutilde_k(x), \Sigmatilde_k(x)$ and $\nutilde_k$ are defined in terms of $(\mu_k, \Sigma_k, \rho_k)$ as in \eqref{eq: tilde defs}.
Let $\psi = \{\psi_1, \dots, \psi_G\}$ and $T_k(z|x; \psi_k) = \int  T_k(z|x; \psi_k, \rho_k)\lambda_k(d\rho_k) $.
\begin{lemma} [Mixture of $t$ transition densities]\label{eq: mixt t transition densities}
Let
\begin{align}\label{eq: mixture of t}
g_M(z; \psi) & = \sum_{k=1}^G \omega_k\zeta_k(z; \psi_k) \text{   and  }   T_{g_M}^\text{CMH}(z|x;\psi) = \sum_{k=1}^G \omega (k|x)  T_k(z|x;\psi_k),
\end{align}
where $\omega (k|x)  = \omega_k \zeta_k(x; \psi_k)/g_M(x; \psi)$.
Then,
\begin{itemize}
\item [(i)] $g_M(z; \psi)$ is a mixture of $t$ densities and $T_{g_M}^\text{CMH}(z|x;\psi)$ is a mixture of $t$ transition densities with $T_{g_M}^\text{CMH}(z|x;\psi)$  a reversible transition density with invariant $g_M(z; \psi)$;
    \item [(ii)]
if the proposal density is $T_{g_M}^\text{CMH}(z|x;\psi)$ and the target density is $\pi(z)$, then the \MH{} acceptance probability is
\begin{align}\label{eq: acc prob t mix}
\alpha(z,x) & = \min \left \{ 1, \frac{\pi(z)}{\pi(x)} \frac{ g_M(x; \psi)}{ g_M(z; \psi)} \right\}.
\end{align}
\end{itemize}
\end{lemma}

We note that it is straightforward to generate from $T_{g_M}^\text{CMH}(z|x;\psi)$, given $x$, because it is a mixture of transition densities, each of which is a mixture. However, it is difficult to compute $T_{g_M}^\text{CMH}(z|x;\psi)$ because it is difficult to compute each of $T_k(z|x;\psi_k)$ as it is  a $t$ density mixed over $\rho_k$. However, by Part (ii) of Lemma~\ref{eq: mixt t transition densities},
it is straightforward to compute the acceptance probability $\alpha(z,x)$.

\subsection{Constructing reversible mixtures of conditional $t$ densities}\label{subsec:conditional}
Suppose that the vector $z$ has density $g_M(z; \psi)$ which is a mixture of multivariate $t$ densities as
in equation \eqref{eq: mixture of t}, with $\zeta_k(z; \psi_k) = t_d(z; \mu_k, \Sigma_k, \nu_k)$.
We partition $z$ as $z = (z_A, z_B)$, where $z_A $ is $d_A \times 1$ and we partition the $\mu_k$ and $ \Sigma_k$ conformally, as
\begin{align*}
\mu_k & = \begin{pmatrix}
\mu_{k,A} \\
\mu_{k,B}
\end{pmatrix}
\quad \text{and} \quad \Sigma_{k} = \begin{pmatrix} \Sigma_{k,AA} & \Sigma_{k,AB}\\
\Sigma_{k,BA} & \Sigma_{k,BB}
\end{pmatrix} , \quad k = 1, \dots, G .
\end{align*}
Then, $\zeta_k(z_A|z_B; \psi_k) = t_{d_A}(z_A; \mutilde_k(z_B), \Sigmatilde_k, \nutilde_k) $, where
$\mutilde_k(z_B) = \mu_{k,A}+\Sigma_{k,AB}\Sigma_{k, BB}^{-1}(z_B-\mu_{k,B})$, $\Sigmatilde_k = \Sigma_{k,AA}-\Sigma_{k,AB}\Sigma_{k, BB}^{-1}\Sigma_{k, BA},$ and
$\nutilde_k = \nu_k+d_A$.
\begin{lemma} [Reversible mixture of conditional densities] \label{lemma: rev cond densities}
Define the transition kernel
\beq\label{eq: cond t mixture}
T_{g_M}^\text{BS}(z|x; \psi) = \sum_{k=1}^G  \frac{\omega_k \zeta_k(x; \psi_k)  }{ g_M(x; \psi)} \zeta_k(z_A|z_B; \psi_k)I(z_B = x_B).
\eeq
Then $T_{g_M}^\text{BS}(z|x; \psi)$ is reversible with invariant density $g_M(z;\psi) $.
\end{lemma}

%======================================================================%
\section{The adaptive correlated \MH{} (ACMH) sampler}\label{Sec:adaptive}
%======================================================================%
The way we  implemented the ACMH sampler is now
described, although Sections~\ref{SS: convergence results} and \ref{Sec:constructing} alow us to construct the sampling scheme in a number of ways.
Section \ref{Sec:estimating mixt} outlines the Variational Approximation method for estimating mixtures of multivariate $t$ distributions.
Sections \ref{Sec:blocking} and \ref{Sec:random} discuss component-wise sampling and adaptive random walk sampling.
Section \ref{Sec:ACMH} summarizes the ACMH sampler.

%----------------------------------------------------------------------%
\subsection{Estimating mixtures of multivariate $t$ densities}\label{Sec:estimating mixt}
%----------------------------------------------------------------------%
Given a mixture of multivariate $t$ densities, Section \ref{subsec:correlated} describes a method to construct reversible mixtures of $t$.
This section outlines a fast Variational Approximation method for estimating such a mixture of $t$.

Suppose that $p(D|\theta)$ is the likelihood computed under the assumption that the data generating process of $D$ is a
mixture of $t$ density $m(x|\theta)$, with $\theta$ its parameters.
Let $p(\theta)$ be the prior, then Bayesian inference is based on the posterior $p(\theta|D)$,
which is often difficult to handle.
Variational Approximation approximates this posterior
by a more tractable distribution $q_\text{va}(\t)$
by minimizing the Kullback-Leibler divergence
\begin{align*} % \label{eq: KL div}
\KL(q_\text{va}) & = \int \log \frac{q_\text{va}(\theta)}{p(\theta|D)} q_\text{va}(\theta) d\theta \ ,
\end{align*}
among some restricted class of densities $q_\text{va}\in\mathcal Q=\{q_\text{va}(\cdot|\lambda),\ \lambda\in\Lambda\}$.
Because,
\begin{eqnarray}
\log\; p(D)  & = & \int \log \frac{p(\theta)p(D|\theta)}{q_\text{va}(\theta)}q_\text{va}(\theta)d\theta
+\int \log \frac{q_\text{va}(\theta)}{p(D|\theta)}q_\text{va}(\theta) d\theta, \label{keyidentity}
\end{eqnarray}
minimizing $\KL(q_\text{va})$ is equivalent to maximizing
\begin{eqnarray}
L(\lambda) &=& \int \log \frac{p(\theta)p(D|\theta)}{q_\text{va}(\theta|\lambda)}q_\text{va}(\theta|\lambda)d\theta.   \label{lowerbd}
\end{eqnarray}
Because of the non-negativity of the Kullback-Leibler divergence term in (\ref{keyidentity}),
\eqref{lowerbd} is a lower bound on $\log\; p(D)$.
We refer the reader to \cite{Tran:2012} who describe in detail how to fit mixtures of $t$ using Variational Approximation,
in which the number of components is automatically selected using the split and merge algorithm by maximizing the lower bound.
The accuracy of Variational Approximation is experimentally studied in \cite{Nott:2012}.
See also \cite{Corduneanu:2001} and \cite{McGrory:2007} who use Variational Approximation for estimating mixtures of normals.

Denote by $\widehat\lambda$ the maximizer of \eqref{lowerbd},
the posterior $p(\theta|D)$ is approximated by $q_\text{va}(\theta|\widehat\lambda)$.
From our experience, the estimate $q_\text{va}(\theta|\widehat\lambda)$
often has a small tail, but it can quickly locate the mode of the true posterior $p(\theta|D)$.
In our context, $m(x|\widehat\theta)$
with $\widehat\theta$ the mode of $q_\text{va}(\theta|\widehat\lambda)$
is used as the mixture of $t$ in the ACMH sampler.

We now explain more fully the main difference between the Variational Approximation approach to constructing proposal densities
of \cite{Freitas:2001} and our approach.
\cite{Freitas:2001} estimate $\pi(x)$ directly as $\pihat_{\vb} (x)$ using Variational Approximation,
i.e. $\pihat_{\vb}$ minimizes
\beqn
\KL(\pi_\text{va})=\int\log\frac{\pi_\text{va}(x)}{\pi(x)}\pi_\text{va}(x)dx
\eeqn
among some restricted class of densities, such as normal densities.
\cite{Freitas:2001} then use $\pihat_{\vb} (x)$ to form the fixed proposal density.
The estimate $\pihat_{\vb} (x) $ often has much lighter tails than $\pi(x)$ \citep[see, e.g.,][]{Freitas:2001},
therefore such a direct use of Variational Approximation estimates for the proposal density in MCMC can be problematic.
Another problem with their approach is that $\pihat_{\vb} (x) $ needs
to be derived afresh for each separate target density $\pi(x)$ and this may be difficult for some
targets.
In our approach
we act as if the target density $\pi(x)$ is a $t$ mixture $m(x|\theta)$ with parameters $\theta$,
and obtain a point estimate of $\theta$ using Variational Approximation.
The approach is general because it
is the same for all targets, and does not suffer from the problem of light tails.

%----------------------------------------------------------------------%
\subsection{Metropolis within Gibbs component-wise sampling}\label{Sec:blocking}
%----------------------------------------------------------------------%
In high dimensions, generating the whole proposal vector at
the one time may lead to a large difference between the values of the target
$\pi$,  and the proposal $q$, at the proposed and current states.
This may result in high rejection rates in the Metropolis-Hastings
algorithm, making it hard for the sampling scheme to move. To
overcome this problem, we use Metropolis within Gibbs component-wise sampling
in which the coordinates of $x$ are divided into two or more components at each iterate.
Without loss of generality, it is only necessary to consider two components,
a component  $x_B$ that remains unchanged
and a complementary component  $x_A$
generated conditional on $x_B$.
We will refer to $B$ and $A$ as the index vectors of $x_B$ and $x_A$ respectively.
Let $d_B$ and $d_A$ be the dimensions of $B$ and $A$, respectively.
We note that $B$ and $A$ can change at random or systematically
from iteration to iteration. See \cite{Johnson:2011} for a further
discussion of Metropolis within Gibbs sampling.

We can carry out a Metropolis within Gibbs sampling step based on
reversible mixtures of conditional $t$ distributions as in Section~\ref{subsec:conditional}.

In some applications, there are natural groupings of the parameters,
such as the group of mean parameters and the group of variance parameters.
Otherwise, the coordinates $x_B$ can be selected randomly.
For example, each coordinate is independently included in $B$ with probability $p_B$.
The number of  coordinates $d_A$ in $x_A$ should be kept small
in order for the chain to move easily.
We find that it is useful to set $p_B$ so that the expected value of $d_A$ is about 10, i.e.,
 $p_B\approx 1-10/d$.

%----------------------------------------------------------------------%
\subsection{The adaptive random walk \MH{} proposal}\label{Sec:random}
%----------------------------------------------------------------------%
Using mixtures of $t$ for the proposal distribution helps to quickly and efficiently locate the modes
of the target distribution.
In addition, it is useful to add some random walk \MH{} steps
to explore more effectively the tail regions around the local modes; see Section~\ref{Exa:mixture}.

We use the following version of an adaptive random walk step,
which takes into account the potential multimodality of the target.
Let $x$ be the current state and $g_M(x; \psi)$  the latest mixture of $t$ as in \eqref{eq: mixture of t}.
Let $\khat(x; \psi) =\argmax_k\{\omega_k \zeta_k(x; \psi_k)\} $, %
i.e. $\khat (x; \psi) $ is the index of the component of the mixture that $x$ is most likely to belong to.
Let $\phi_d(x; a, B)$ be a $d$-variate normal density with mean vector $a$ and covariate matrix $B$.
The  random walk proposal density is $q^{RW}(z|x; \khat (x; \psi) ) = \phi_d (z; x, \kappa \Sigmatilde _{\khat(x; \psi) })$,
where $\Sigmatilde _{\khat(x; \psi)}  = \nu_{\khat(x; \psi) }/(\nu_{\khat(x; \psi) }-2)\Sigma_{\khat(x, \psi)}$%
if $\nu_{\khat(x; \psi)}>2$ and is equal to $\Sigma_{\khat(x; \psi)}$ otherwise.
The scaling factor $\kappa  = 2.38^2/d$ \citep[see][]{Roberts:2009}.

%----------------------------------------------------------------------%
\subsection{Description of the ACMH sampler}\label{Sec:ACMH}
%----------------------------------------------------------------------%
This section gives the details of the ACMH sampler,
which follows the framework in Section \ref{Sec:theory} to ensure convergence.
The sampler consists of a reversible proposal density together
with a random walk proposal. We shall first describe the reversible proposal density.
Let $g_0(z)$ be the heavy tailed component and $g_M(z; \psihat(\calHn))$ the mixture of $t$ densities
described in \eqref{eq: mixture of t}, where $\calHn$ is the history vector obtained after iteration $n$ based on the trial chain
$X^\prime $ and $\psihat$ is the
estimate of $\psi$ based on $\calHn$.
Let $T_{g_0} (z|x) = g_0(z)$ be the reversible transition density whose invariant
density is $g_0(z)$ (see Part (i) of
Lemma~\ref{lemma: properties of reversible densities}  ). Let
$T_{g_M}^\text{CMH}(z|x; \psihat(\calHn) ) $ be the correlated reversible transition density  defined  in equation
\eqref{eq: mixture of t} and  let $T_{g_M}^\text{BS}(z|x; \psihat(\calHn) ) $ be the component-wise mixture reversible transition
density  defined in \eqref{eq: cond t mixture}. We now define the mixtures,
\begin{align*}
T_{g_M} (z|x; \psihat(\calHn)) & = (1-\gamma) T_{g_M}^\text{CMH}(z|x; \psihat(\calHn) ) +
\gamma T_{g_M}^\text{BS}(z|x; \psihat(\calHn) )\\
q^\ast (z; \psihat(\calHn)) & = \beta_0 g_0(z) + (1- \beta_0) g_M(z; \psihat(\calHn) ) \\
T_{q^\ast} (z|x; \psihat(\calHn)) & = \frac{\beta_0 g_0(x)} { q^\ast (x; \psihat(\calHn))} T_{g_0}(z|x) + \frac{(1-\beta_0)g_M(x; \psihat(\calHn) )  }{q^\ast (x; \psihat(\calHn))} T_{g_M} (z|x,; \psihat(\calHn)),
\end{align*}
and
\begin{align}\label{e:proposal}
q(z|x; \psihat(\calHn)) & = \delta q^\ast (z; \psihat(\calHn)) + (1-\delta ) T_{q^\ast} (z|x; \psihat(\calHn))\notag \\
& = \delta \beta_0 g_0(z) + \delta (1-\beta_0) g_M(z; \psihat(\calHn) ) + (1-\delta)\beta_0 \frac{g_0(x)}{q^\ast (x; \psihat(\calHn))} T_{g_0}(z|x) + \\ & + (1-\delta) ( 1- \beta_0) \frac{ g_M(x; \psihat(\calHn)   }  {q^\ast (x; \psihat(\calHn))} T_{g_M} (z|x; \psihat(\calHn)),
\end{align}
with $0 \leq \gamma, \beta_0 , \delta \leq 1$.
Note that $\delta$ is the probability of generating an independent proposal
and $\gamma$ is related to the probability of doing component-wise sampling.
Then,
%------------------------------%
\begin{lemma} \label{lemma: acmh proposal}
%------------------------------%
\begin{itemize}
\item [(i)] $T_{g_M} (z|x; \psihat(\calHn)) $ is a reversible Markov transition density with invariant density $g_M(z; \psihat(\calHn))$.
\item [(ii)] $T_{q^\ast} (z|x; \psihat(\calHn)) $ is a reversible Markov transition density with invariant density $q^\ast (z; \psihat(\calHn)) )$.
\item [(iii)] $q(z|x; \psihat(\calHn))$ is a reversible Markov transition density with invariant density $q^\ast (z; \psihat(\calHn)) )$.
\item [(iv)] If $q(z|x; \psihat(\calHn))$ is a proposal density with target density $\pi(z)$, then
the acceptance probability is
\begin{align*}
\alpha(z, x; \psihat(\calHn)) & = \min\left  \{ 1, \frac{\pi(z)   } { \pi(x) }
 \frac{q^\ast (x; \psihat(\calHn)) }{q^\ast (z; \psihat(\calHn)) }\right \}.
\end{align*}
\end{itemize}
\end{lemma}

\paradot{Description of the ACMH proposal density}
The ACMH sampler consists of a reversible proposal density together
with a random walk proposal.
Let $p_n(x_{n-1},dx_n)$ be the transition kernel at iteration $n$ of the main chain $X$.
Denote by $p_{1,n}$, $p_{2,n}$ the transition kernel with respect to
the reversible proposal $q(z|x; \psihat(\calHn))$ and the random walk proposal $q^{AR}(z|x; \khat (x; \psihat(\calHn)) )$ respectively.
\begin{itemize}
\item [(1)] $p_n = p_{1,n}p_{2,n}$ at $n = \iota_{RW},\ 2\iota_{RW}, ...$ (see Corollary \ref{corollary} (iv)).
That is, a composition of a correlated \MH{} step with reversible proposal and a random walk step is performed after every $\iota_{RW}-1$ iterations.
In our implementation we take $\iota_{RW}=10$.
\item [(2)]
In all the other steps, we  take  $p_n = p_{1,n}$.
\end{itemize}
%-----------------------------------------------%
\paradot{Convergence of the ACMH sampler}
%-----------------------------------------------%
If we choose $g_{0}(z)$ such that $g_{0}(z)\geq \beta_0\pi(z)$
for some $0<\beta_0<1$, then $q(z|x; \psihat(\calHn))\geq \delta\beta_{0}=\beta$.
By Corollary \ref{corollary}, we have a formal justification of the convergence of the ACMH sampler.

%-------------------------%
\begin{corollary}\label{corollary:ACMH_convergence}
%-------------------------%
Suppose that
\beq\label{e:bounded1}
g_0(z)\geq\beta_0\pi(z)\;\;\text{for all}\;\; z\in E,
\eeq
for some $0<\beta_0<1$.
Then Theorems \ref{ergodic_theorem} and \ref{SLLN_theorem} hold for the ACMH sampler for any history $\H^n$.
\end{corollary}

For a general target $\pi(z)$, $g_0(z)$ can be informally selected such that it is sufficiently heavy-tailed
to make \eqref{e:bounded1} hold.
In Bayesian inference, $\pi(z)$ is a posterior density that is proportional to $p(y|z)p(z)$ with $p(z)$ the prior
and $p(y|z)$ the likelihood.
Suppose that the likelihood is bounded; this is the case if the maximum
likelihood estimator exists.
If $p(z)$ is a proper density and we can generate from it, then we can set $g_0(z)=p(z)$ and it is straightforward to check that the condition \eqref{e:bounded1} holds.
The boundedness condition \eqref{e:bounded1} is satisfied
in all the examples in this paper.

We now
briefly discuss the cost of running the two chain algorithm as in our
article, compared to running a single chain adaptive algorithm. If the
target density is inexpensive to evaluate, then the cost of running the two
chain sampler is very similar to the cost of running just one chain because
the major cost is incurred in updating the proposal distribution. If it is
expensive to evaluate the target, then we can run the two chains in parallel
on two (or more) processors. This is straightforward to do in programs such
as Matlab because multiple processors are becoming increasingly common on
modern computers.

Section \ref{Sec:initialization} discusses the initial proposal $q_0$, the history vector $\mathcal{H}^0$ and $g_0$.

%-----------------------------------------------%
\subsubsection{Two-stage adaptation}\label{subsec:first stage}
%-----------------------------------------------%
We run the adaptive sampling scheme in two stages.
Adaptation in the first stage is carried out more intensively by
re-estimating the mixture of $t$ distributions after every 2000 iterations,
and then every 4000 iterations in the second stage.
When estimating the mixtures of $t$ in the first stage,
we let the Variational Approximation algorithm determine the number of components.
While in the second stage, we fix the number of components at that number in the mixture obtained after the first stage.
This makes the procedure faster and helps to stabilize the moves.
In addition, it is likely that the number of components is unchanged in this second stage.

%-----------------------------------------------%
\subsubsection{Selecting the control parameters}\label{ssub:corr}
%-----------------------------------------------%
When the mixture of $t$ approximation $g_M(z;\widehat\psi(\calHn))$ becomes closer to the target,
we expect the proposal $q(z|x; \psihat(\calHn))$ to be close to $g_M(z;\widehat\psi(\calHn))$.
We can do so by setting $\delta=\delta_n\to1$ as $n$ increases and setting a small value to $\beta_0$.
In our implementation we take $\beta_0=0.001$, $\gamma=0.2$ and a sequence $\delta_n$ as follows.
Let $N$ be the length of the chain we wish to generate
and suppose that $N = a_Nb_N$.
We set $\delta_{kb_N+j}=(k+1)/a_N$ for $k=0,...,a_N-1$ and $j=1,...,b_N$.
In our implementation we take $a_N = 10$.
For the correlation parameter $\rho$, we simply select the probability measure $\lambda(\rho)$ as the
$\text{Beta}(1,1)$ distribution.
These values were set after some experimentation.
However, it is likely that we can further improve the efficiency of the ACMH sampler with
a more careful (and possibly adaptive) choice of these control parameters.

%-----------------------------------------------%
\section{Initial exploration}\label{Sec:initialization}
%-----------------------------------------------%
The purpose of the ACMH sampler is to deal with non-standard and multimodal target distributions.
The sampler works more efficiently if the adaptive chain starts from
an initial mixture distribution that is able to roughly locate the modes.
We therefore attempt to initialize the history vector $\mathcal{H}^0$ by a few draws generated approximately from $\pi$
by an algorithm that can explore efficiently the whole support of the target,
and then estimate the initial mixture of $t$ based on these draws.
Our paper uses simulated annealing \citep{Neal:2001} to initialize the sampler.
An alternative is to use the Wang-Landau algorithm \citep{Wang:2001a,Wang:2001b}.
However, this algorithm requires the user to partition the parameter space appropriately
which is difficult to do in many applications.

\paradot{Simulated annealing}
Simulated annealing works by moving from an easily-generated distribution to the distribution
of interest through a sequence of bridging distributions.
Annealed sampling has proved useful in terms of efficiently exploring the support
of the target distribution \citep{Neal:2001}.
Let $\pi_0(x)$ be some easily-generated distribution, such as a $t$ distribution,
and $\psi_t,\ t=0,1,...,T$ a sequence of real numbers such that $0=\psi_0<...<\psi_T=1$.
A convenient choice is $\psi_t=t/T$. Let
\beqn
\eta_t(x)=\pi_0(x)^{1-\psi_t}\pi(x)^{\psi_t}.
\eeqn
Note that $\eta_0$ is the initial distribution $\pi_0$ and $\eta_T$ is the target $\pi$.
We sample from this sequence of distributions
using the sequential Monte Carlo method \citep[see, e.g.][]{DelMoral:2006,Chopin:2004}, as follows.
\begin{enumerate}
\item Generate $x_i\sim\eta_0(\cdot),\ i=1,...,N_p$, where $N_p$ is the number of particles.
\item For $t=1,...,T$
\begin{itemize}
\item[(i)] Reweighting: compute the weights
\beqn
\widetilde w_i=\frac{\eta_t(x_i)}{\eta_{t-1}(x_i)},\ \ \ w_i=\frac{\widetilde w_i}{\sum_{j=1}^{N_p}\widetilde w_j}.
\eeqn
\item[(ii)] Resampling: sample from $(x_i,w_i)_{i=1,...,N_p}$ using stratified sampling. Let $(\widetilde x_i)_{i=1,...,N_p}$ be the resampled particles.
\item[(iii)] Markov move: for $i=1,...,N_p$, generate $x_i^{(m)}\sim P_{\eta_t}(\cdot|x_i^{(m-1)})$, $m=1,...,M$, where $P_{\eta_t}(\cdot|\cdot)$ is a Markov kernel with invariant distribution $\eta_t$, $x_i^{(0)}=\widetilde x_i$.
$M$ is the burnin number.
\item[(iv)] Set $x_i\leftarrow x_i^{(M)},\ i=1,...,N_p$.
\end{itemize}
\end{enumerate}
The above sequential Monte Carlo algorithm produces particles $x_i$ that are approximately generated from the target $\pi$ \citep{DelMoral:2006}.
We can now initialize the history vector $\mathcal{H}^0$ using these particles and $q_0$ by the mixture of $t$ estimated from $\mathcal{H}^0$.
Typically, $T$ should take a large value for multimodal and high-dimensional targets.
In the default setting of the ACMH sampler, we set $T=10$, $N_p=500$ and $M=10$.
The initial distribution $\pi_0$ is a multivariate $t$ distribution with location $\mu_0=(0,...,0)'$, scale matrix $\Sigma_0=I_d$ and 3 degrees of freedom.
However, it is useful to estimate $\mu_0$ and $\Sigma_0$ from a short run of an adaptive random walk sampler,
and we follow this approach in the real data examples.

In the default setting of the ACMH sampler, we select the heavy-tailed component $g_0(z)$ as $q_0(z)$
except that all the degrees of freedom of the $t$ component of $q_0$ are set to 1,
so that the boundedness condition \eqref{e:bounded1} is likely to be satisfied.
However, in all the examples below, $g_0$ is context-specified to make sure that \eqref{e:bounded1} holds.

%======================================================================%
\section{Simulations}\label{Sec:simulations}
%======================================================================%
A common performance measure for an MCMC sampler
is the integrated autocorrelation time (IACT).
For simplicity, consider first the univariate case
and let $\{x_i,i=1,...,M\}$ be the generated iterates from the Markov chain.
Then the IACT is defined as
\beqn
\text{IACT} = 1+2\sum_{t=1}^\infty\rho_t,
\eeqn
where $\rho_t=\text{corr}(x_1,x_{t+1})$ is the autocorrelation of the chain at lag $t$.
Provided that the chain has converged,
the mean $\mu$ of the target distribution is estimated by $\bar x=\sum_i x_i/M$ whose variance is
\beqn
\text{Var}(\bar x)=\frac{\sigma^2}{M}\left(1+2\sum_{t=1}^{M-1}\Big(1-\frac tM\Big)\rho_t\right)\approx\frac{\sigma^2}{M}\Big(1+2\sum_{t=1}^\infty\rho_t\Big)=\text{IACT}\cdot\frac{\sigma^2}{M},
\eeqn
where $\sigma^2$ is the variance of the target distribution.
This shows that the IACT can be used as a measure of performance and that the smaller the IACT, the better the sampler.
Following \cite{Pitt:2012b}, we estimate the IACT by
\beqn
\widehat{\text{IACT}}=1+2\sum_{t=1}^{L^*}\widehat\rho_t,
\eeqn
where $\widehat\rho_t$ are the sample autocorrelations,
and $L^*=\min\{1000,L\}$,  with $L$ the first index $t$ such that $|\widehat\rho_t|\leq2/\sqrt{K_t}$
where $K_t$ is the sample size used to estimate $\widehat\rho_t$.
That is, $L$ is the lowest index after which the estimated autocorrelations are randomly scattered about 0.
When $d>1$  we take, for simplicity, the average IACT over the $d$ coordinates, or the maximum IACT.

Another performance measure is the squared jumping distance, \cite[see, e.g.,][ and the references in that paper]{Pasarica:2003}.
For the univariate case,
\begin{align*}
\text{Sq distance} & = \frac{1}{N-1}\sum_{i=1}^N |x_{i+1}-x_i|^2\approx 2\sigma^2(1-\rho_1).
\end{align*}
Therefore, the larger the squared distance the better.
When $d>1$, we take the average squared distance or the minimum squared distance over the $d$ coordinates.
We also report the acceptance rates in the examples below.

The IACT and squared distance are good performance measures when the target is unimodal.
If the target is multimodal, these measures may not be able to determine whether or not the chain has converged to the target, as discussed below.
We introduce another measure which suits the context of a simulation example
where a test data set $\mathcal D_T=\{x_s=(x_{s1},...,x_{sd})', s=1,...,S\}$ generated from the target $\pi$ is available.
Let $\hat f_i$ be the kernel density estimate of the $i$th marginal $\pi_i$ of the target.
The Kullback-Leibler divergence between $\pi_i$ and ${\hat f}_i$ is
\beqn
\int\log \left ( \frac{\pi_i(x_i)}{\hat f_i(x_i)}\right ) \pi_i(x_i)dx_i\approx C_i-\LPDS_i,
\eeqn
where $C_i=\int \pi_i(x_i)\log\;\pi_i(x_i)dx_i$ is independent of $\hat f_i$ and
\beqn
\LPDS_i = \frac1S\sum_{s=1}^S\log\hat f_i(x_{si})\approx \int\pi_i(x_i)\log\hat f_i(x_{i})dx_i
\eeqn
is the log predictive density score for the $i$th marginal.
Clearly, the bigger the $\LPDS_i$, the closer the estimate $\hat f_i$ to the true marginal $\pi_i$.
We define the log predictive density score over the $d$ marginals by
\beqn
\LPDS=\frac{1}{d}\sum_{i=1}^d\LPDS_i.
\eeqn
The bigger the \lpds{}, the better the MCMC sampler.

%------------------------------------------%
\subsection{Target distributions}\label{Exa:mixture}
%------------------------------------------%
The first target is a mixture of two multivariate skewed normal distributions
\beq\label{e:mixture target}
\pi_\text{msn}(x) = \sum_{k=1}^2\varphi_k\mathcal{SN}_d(x;\mu_k,\Sigma_k,\lambda_k),
\eeq
where $\mathcal{SN}_d(x;\mu,\Sigma,\lambda)$ denotes the density of a $d$-dimensional skewed normal distribution
with location vector $\mu$, scale matrix $\Sigma$ and shape vector $\lambda$.
See, e.g., \cite{Azzalini:1999}, for an introduction to multivariate skewed normal distribution.
We set $\mu_1=(-5,...,-5)'$, $\mu_2=(5,...,5)'$, $\Sigma_1=\Sigma_2=5(\sigma_{ij})_{i,j}$ with $\sigma_{ij}=(-0.5)^{|i-j|}$,
$\lambda_1=(-10,...,-10)'$, $\lambda_2=(10,...,10)'$
and $\varphi_1=0.6$, $\varphi_2=0.4$.
It is straightforward to sample directly and exactly from a skewed normal \citep{Azzalini:1999},
and therefore from $\pi_\text{msn}(x)$.
However, this is a non-trivial problem for MCMC simulation, especially in higher dimensions, because the target is multimodal
with an almost-zero probability valley between the two modes;
see the left panel of Figure \ref{fig:banana_mixture_plot} for a plot of $\pi_\text{msn}(x)$ when $d=2$.

Let $f_k(x)=\mathcal{SN}_d(x;\mu_k,\Sigma_k,\lambda_k)$.
By the properties of the multivariate skewed normal distribution \citep[see,][]{Azzalini:1999},
$f_k(x)\leq 2g_k(x)$, where $g_k(x)=N_d(x;\mu_k,\Sigma_k)$ is the density of the $d$-dimensional
normal distribution with mean $\mu_k$ and covariance matrix $\Sigma_k$.
The boundedness condition \eqref{e:bounded1} is satisfied by setting $g_0(z)=\varphi_1g_1(z)+\varphi_2g_2(z)$,
because then $\pi(z)/g_0(z)$ is bounded.

The second target density is the banana-shaped
distribution considered in \cite{Haario:1999}
\beq\label{e:banana}
\pi_b(x) = N_d(\phi_b(x);0,\Sigma),
\eeq
where $\Sigma = \diag(100,1,...,1)$, $\phi_b(x) = (x_1,x_2+bx_1^2-100b,x_3,...,x_d)$
and $b=0.03$.
See the right panel of Figure \ref{fig:banana_mixture_plot} for a plot of $\pi_b(x)$ when $d=2$.
As shown, the banana-shaped density has a highly non-standard support with very long and narrow tails.
It is challenging to sample from this target \citep{Haario:1999,Haario:2001,Roberts:2009}.

It can be shown after some algebra that the first marginal of $\pi_b(x)$
is $N(0,10^2)$ and for $i=3,...,d$ the marginals are independent $N(0,1)$.
It can be visually seen that the support of the second marginal is basically in the interval $(-50,50)$.
We therefore informally impose the condition \eqref{e:bounded1} by selecting $g_0(z)=t_d(z;0,\wt\Sigma,5)$,
a multivariate $t$ density with location 0, scale matrix $\wt\Sigma=\diag(100,100,1,...,1)$ and 5 degrees of freedom.
Typically, this ensures that the support of $g_0(z)$ covers the support of $\pi_b(z)$
and therefore the boundedness condition \eqref{e:bounded1} holds.

\begin{figure*}[h]
\centerline{\includegraphics[width=.8\textwidth,height=.5\textwidth]{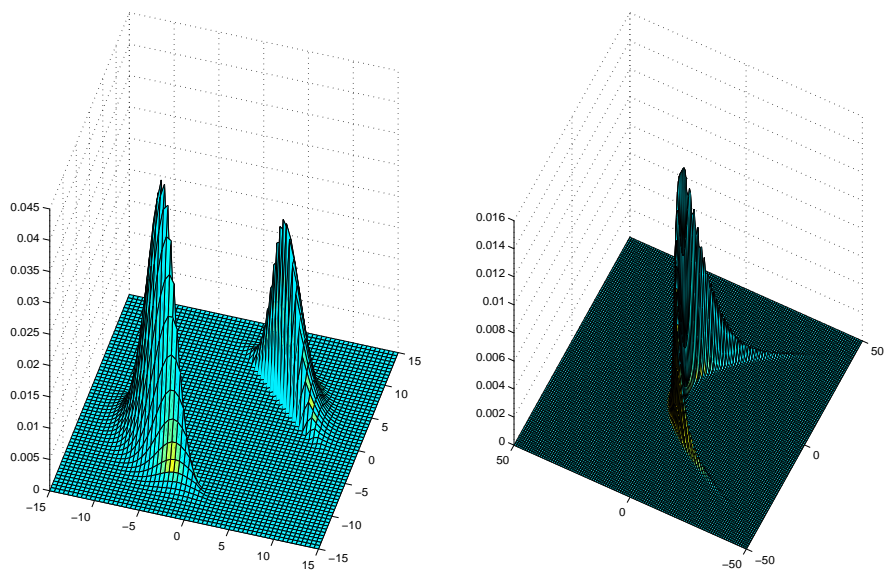}}
\caption{Plots of the probability density functions of the mixture example (left) and the banana-shaped example (right) for $d=2$.
\label{fig:banana_mixture_plot}
}
\end{figure*}

%------------------------------------------%
\subsection{Performance of the ACMH sampler}\label{subsec:perfomacce}
%------------------------------------------%
This section reports the performance of the ACMH sampler and compares it to
the adaptive random walk sampler (ARWMH) of \cite{Haario:2001}
and the adaptive independent \MH{} sampler (AIMH) of \cite{Giordani:2010}.
For all the samplers, we ran 50,000 iterations with another 50,000 for burnin.

%------------------------------------------%
\subsubsection{The usefulness of the adaptive random walk step}
%------------------------------------------%
We first demonstrate the importance of the adaptive random walk step by
comparing the performance of the ACMH sampler to a variant of it that does not perform the random walk step.
To make it easier to see the resulting estimates, we consider the target \eqref{e:mixture target} with $d=1$.
The left panel in Figure \ref{fig:RWimportance}
plots the kernel density estimates of the target estimated from the chains with and without the random walk step,
as well as the true density.
All the kernel density estimation reported in this paper
is done using the built-in Matlab function \texttt{ksdensity} with the default setting.
The right panel also plots the estimated kernel densities of
the first marginal when sampling from the banana-shaped target \eqref{e:banana}.
The first marginal of the banana-shaped target has very long tails (see Figure \ref{fig:banana_mixture_plot})
and it is challenging for adaptive MCMC samplers to efficiently explore the extremes of these tails.
The plots show that the chain with the random walk step explores the tail areas around the local modes more effectively.

We now formally justify the claim above using the censored likelihood scoring rule proposed in \cite{Diks:2011}.
This scoring rule is a performance measure for assessing the predictive accuracy of a density estimator $\wh f(x)$
over a specific region of interest, which is the tail area in our problem.
Let $A$ denote the region of interest, $\mathcal D$ a set of $n$ observations.
Then the censored likelihood score is defined as
\beq\label{e:score}
S(\wh f,\mathcal D)=\frac{1}{n}\sum_{x\in\mathcal D}\left(1_{x\in A}\log\wh f(x)+1_{x\in A^c}\log\int_{A^c}\wh f(z)dz\right),
\eeq
where $A^c$ is the complement of set $A$.
This scoring rule works similarly to the popular logarithmic scoring rule \citep{Good:1952};
in particular the bigger $S(\wh f,\mathcal D)$ is, the better the performance of $\wh f$.
However the censored likelihood score takes into account the predictive accuracy in a particular region of interest;
see \cite{Diks:2011} for a more detailed interpretation.

We consider the case of the mixture target $\pi_\text{msn}(x)$ with $d=1$
and are interested in how efficiently the ACMH samplers, with and without the random walk step, explore the left and right tails of $\pi_\text{msn}(x)$.
Let $\wh f_1(x)$ and  $\wh f_2(x)$ be the kernel densities estimated from the chains with and without the random walk step, respectively.
We compute the score \eqref{e:score} for $\wh f_1$ and $\wh f_2$ based on $n=5000$ independent draws from the target $\pi_\text{msn}(x)$,
in which the tail area is defined as $A=\{x\in\mathbb{R}: x<-15\;\text{  or  }\;x>15\}$.
We replicate the computation 10 times.
The scores averaged over the replications with respect to the ACMH samplers with and without the random walk step
are 0.98 and 0.96 respectively.
This result formally justifies the claim that the random walk step helps the sampler to explore the tail area more effectively.

We also ran long chains with 200,000 iterations after discarding another 200,000 for burnin,
then the difference between the censored likelihood scores of the ACMH samplers with and without the random walk step
is 0.0008. That is, the difference decreases when the number of iterations increases.
This result suggests that the ACMH sampler without the random walk step is able to explore the tail area effectively
if it is run long enough.

\begin{figure*}[h]
\centerline{\includegraphics[width=.9\textwidth,height=.5\textwidth]{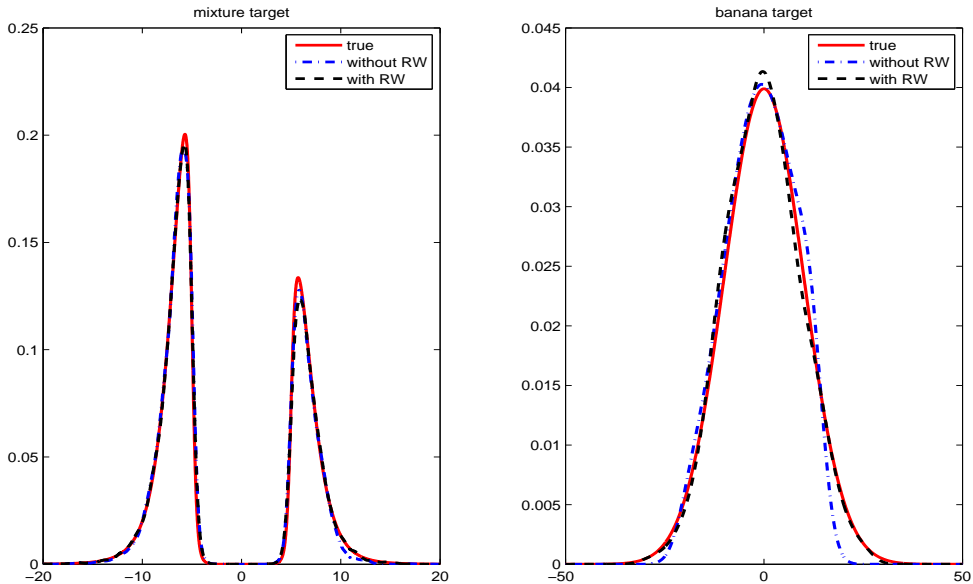}}
\caption{\label{fig:RWimportance}
The left panel plots the kernel density estimates of the mixture target
obtained from the ACMH chains with and without the random walk step, as well as the true density.
Similarly, the right panel plots the densities with respect to the first marginal of the banana-shaped target.
}
\end{figure*}

%------------------------------------------%
\subsubsection{The usefulness of the component-wise sampling step}
%------------------------------------------%
To illustrate the effect of the component-wise sampling step,
we sample from the banana-shaped target using the ACMH samplers with and without this step.
The coordinates $x_B$ that are kept unchanged, and therefore the size $d_B$, are selected randomly as in Section \ref{Sec:blocking}.
Table \ref{t:blocking} summarizes the performance measures for these two samplers averaged over 10 replications.
The result shows that in general the sampler that performs the component-wise sampling step outperforms
the one that does not.

\begin{table}[h]
\centering
\vskip2mm
{\small
\begin{tabular}{ccccc}
\hline\hline
$d$	&Algorithm		&Acceptance rate (\%)&IACT	&Sq distance\\
\hline
10	&Without component-wise sampling	&38		&38.61	&2.63\\
	&With component-wise sampling	    	&56		&19.45	&5.92\\
\hline
20	&Without component-wise sampling	&32		&57.11	&0.71\\
	&	With component-wise sampling    	&34		&47.33	&1.52\\
\hline
40	& 	Without component-wise sampling&14		&171.4	&0.18\\
	&  With component-wise sampling   	&26		&80.17	&0.83\\
\hline\hline
\end{tabular}
}
\caption{Importance of component-wise sampling.
}\label{t:blocking}
\end{table}

\begin{table}[h]
\centering
\vskip2mm
{\small
\begin{tabular}{ccccc}
\hline\hline
$d$	&Algorithm		&Acceptance rate (\%)&IACT	&Sq distance\\
\hline
10	&Without CMH		&35		&26.71	&2.49\\
	&With CMH	    	&56		&19.45	&5.92\\
\hline
20	&Without CMH		&26		&62.11	&0.88\\
	&With CMH	    	&34		&47.33	&1.52\\
\hline
40	&Without CMH		&0.1		&180.2	&0.01\\
	&With CMH	    	&26		&80.17	&0.83\\
\hline\hline
\end{tabular}
}
\caption{Importance of the correlated \MH{} step.
}\label{t:CMH}
\end{table}

%------------------------------------------%
\subsubsection{The usefulness of the reversible \MH{} step}
%------------------------------------------%
We demonstrate the importance of the reversible step
by comparing the ACMH sampler with a version of it in which the $\delta$ parameter in Section 4.4 is set to one,
i.e. only the independent \MH{} step is performed.
Table \ref{t:CMH} summarizes the performance measures for these two samplers averaged over 10 replications.
The results suggest that the correlated step helps improve significantly on the performance of the ACMH sampler.

%------------------------------------------%
\subsubsection{Comparison of the ACMH sampler to other adaptive samplers}\label{SS: comparison}
%------------------------------------------%
We now compare the ACMH sampler to the ARWMH and AIMH samplers
for the mixture of skewed normals and banana-shaped targets.

\paradot{The mixture target}
Figure \ref{fig:mixtureexample1} plots the chains with respect to the first marginal for three cases: $d=2$, $d=5$ and $d=10$.
The ARWMH never converges to the target even when $d=1$.
It looks as if the ARWMH has converged
but in fact it is always stuck at a local mode because the target has two modes that are almost separate.
In such cases, the performance of ARWMH may be mistakenly considered to be good in terms of IACT and squared jumping distance,
while its \lpds{} will be large because the estimated density is far from the true density.
This justifies the introduction of the \lpds{} as a performance measure.
The AIMH sampler works well when $d$ is as small as 3 in this hard example.
As expected with samplers based on independent \MH{} steps only,
it is almost impossible for the AIMH to move the chain when $d$ is large.
Figure \ref{fig:mixtureexample1} shows that the ACMH sampler converges best.

\begin{table}[h]
\centering
\vskip2mm
{\small
\begin{tabular}{cccccc}
\hline\hline
$d$	&Algorithm	&Acceptance rate (\%)&IACT	&Sq distance	&LPDS\\
\hline
2	&ARWMH		&36		&9.15	&1.73		&-73.1\\
	&AIMH	    	&14		&13.6	&10.5		&-2.81\\
	&ACMH	    	&70		&3.97	&37.4		&-2.80\\
\hline
5	&ARWMH 		&32		&15.2 	&1.13		&-42.5\\
	&AIMH		&11		&22.2	&6.12		&-2.89\\
	&ACMH		&68		&3.91	&34.3		&-2.84\\
\hline
10	&ARWMH		&29		&29.5	&0.63		&-43.3\\
	&AIMH		&0.01		&1885	&0.006		&-20.8\\
	&ACMH		&68		&9.40	&40.27		&-2.86\\
\hline
\hline
\end{tabular}
}
\caption{
Mixture of skewed normals target:
The table reports the acceptance rates, autocorrelation times,
squared distances and log predictive density scores for the three
adaptive sampling schemes ARWMH, AIMH, and ACMH and three
values of the dimension $d$.
The values are averaged over 5 replications.
}\label{t:mixture table}
\end{table}

\begin{figure*}[h]
\centerline{\includegraphics[width=.9\textwidth,height=.5\textwidth]{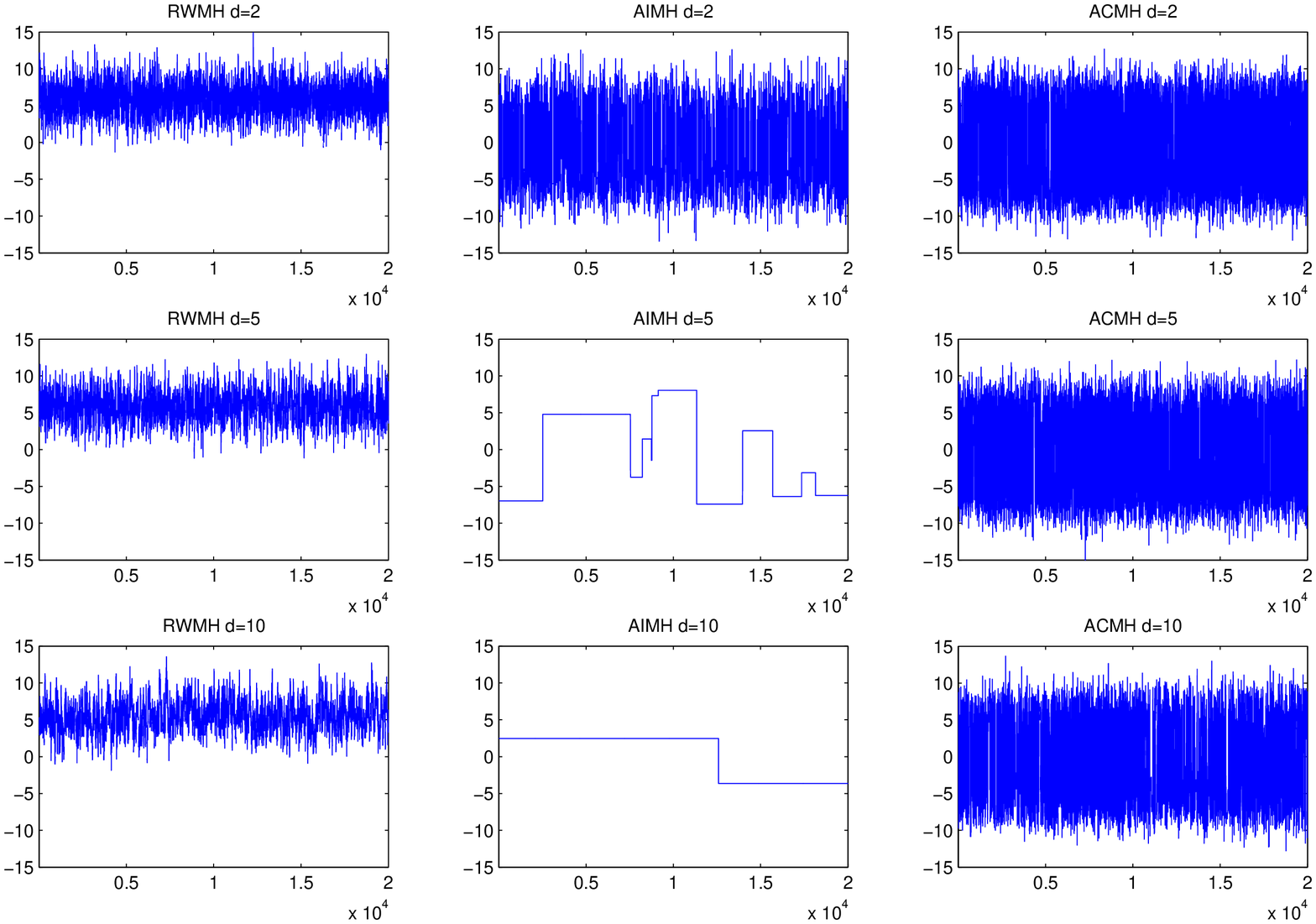}}
\caption{
Mixture of skewed normals target:
Plots of the iterates of the three adaptive
sampling schemes for the first marginal. Columns 1 to 3 correspond
to the ARWMH, AIMH and ACMH sampling schemes. Rows
1 to 3 correspond to the dimensions $d = 2$, 5 and 10.
\label{fig:mixtureexample1}
}
\end{figure*}

We now compare the performance of the three adaptive samplers more formally over 5 replications.
The acceptance rates, IACT, squared distance values and \lpds{s} are computed and averaged over the $d$ marginals and 5 replications.
Table \ref{t:mixture table} summarizes the result.
The ACMH sampler compares favorably with the other two samplers.
It looks as if the ARWMH chain has converged and performs well in terms of IACT and squared distance,
but in fact the chain is trapped in a local mode.
We conclude from these results that if the target is multimodal,
performance measures based on IACT and squared distance may be misleading.

\paradot{Banana-shaped target}
Because this target is unimodal, the IACT and squared distance can be used as performance measures.
Note that we do not use the log predictive
score in this example because it is difficult to obtain
an independent test data set that is generated exactly from the target \eqref{e:banana}.
Table \ref{t:banana table} summarizes the results in terms of percent acceptance rates, IACT, Sq distance and CPU time
in seconds, where the CPU time is the average CPU time over the ten replications. The results are based on 50000 iterations with another 50000 iterations used for burin.
The table also reports two other performance measures that are used by \cite{schmidl:2013}.
The first is II/time = Number of iterates/(IACT $\times$ CPU time), which
is an estimate of the number of independent iterates per unit time, which in this case is seconds.
The second measure Acc/IACT = 1000 $\times $ Acceptance rate/IACT. The table shows that
the ACMH sampler outperforms the other two in terms of Acceptance rate, IACT, Sq distance, Acc/IACT. However, it is worse than the other two samplers in terms of II/time as it takes longer to run.
The code was written in Matlab and run on an Intel Core i7 3.2GHz desktop  running on a single
processor. The time taken by the ACMH sampler can be reduced appreciably by taking the following steps.
(i)~First, by profiling the code we find that
a major part of the time to run the sampler is taken by the variational approximation procedure
used to obtain the mixture of $t$ proposal.
The running time can be shortened appreciably by write the variational approximation part of the code in C or Fortran and using mex files.
(ii)~Second, the running time can also be shortened by running the two chains on separate processors in parallel and also by using parallel
processing for the independent draws.

\begin{table}[!h]
\centering
\vskip2mm
{\small
\begin{tabular}{ccccccccc}
\hline\hline
$d$	&Algorithm	&Acceptance rate (\%)	&IACT	&Sq distance	&CPU time 	& II/time& Acc/IACT\\
\hline
5	&ARWMH		&14			&81.83	&1.23		&21 		&29.1 	 &171 \\
	&AIMH		&20			&44.52	&5.80		&14 		&80.2    &449    \\
	&ACMH		&64			&24.33	&17.0		&253  		&8.12	 &2631\\
\hline
10	&ARWMH		&15			&150.1	&0.41		&22  		&15.1	 &100 \\
	&AIMH		&31			&49.65	&3.06		&15 		&67.1    &624 \\
	&ACMH		&56			&19.45	&5.92		&250  		&10.3 	 &2879\\
\hline
20	&ARWMH		&18			&168.8	&0.15		&26  		&11.4 	 &107 \\
	&AIMH		&10			&174.6	&0.33		&16 		&17.9  	 &57  \\
	&ACMH		&34			&47.33	&1.52		&368 		&2.87    &718 \\
\hline
40	&ARWMH		&24			&208	&0.05		&40  		&6.01    &115\\
	&AIMH		&0			&1991	&0		&16  		&1.57    &0 \\
	&ACMH		&26			&80.17	&0.83		&395  		&1.58    &324\\
\hline
\hline
\end{tabular}
}
\caption{Banana-shaped example: performance measures averaged over 10 replications.
II/sec = Number of iterates/(IACT $\times$ CPU time) and Acc/IACT = 1000 $\times $ Acceptance rate/IACT.
}\label{t:banana table}
\end{table}

%======================================================================%
\section{Applications}\label{Sec:applications}
%======================================================================%
\subsection{Covariance matrix estimation for financial data}
This section applies the ACMH sampler to estimate the covariance matrix
of ten monthly U.S. industry portfolios returns.
The data is taken from the Ken French data library
and consists of $N=990$ observations $\v y=\{y_i,i=1,...,N\}$ from July 1925 to December 2008.
We use the following 10 industry portfolios: consumer non-durable, consumer durable,
manufacturing, energy, business equipment, telephone and television transmission, shops,
health, utilities, and others.

We assume that the $y_i,\ i=1, \dots, N,$ are independently distributed as
$N_p(y_i;0, \Sigma) $ with $p=10$.
We are interested in Bayesian inference of $\Sigma$.
\cite{Yang:1994} propose the following reference prior for $\Sigma$,
\beqn
p_\text{ref}(\Sigma)\propto\frac{1}{|\Sigma|\prod_{i<j}(r_i-r_j)},
\eeqn
where $r_1\geq r_2\geq...\geq r_p$ are the eigenvalues of $\Sigma$.
This reference prior puts more mass on covariance matrices having eigenvalues that are close to each other.
That is, it shrinks the eigenvalues in order to produce a covariance matrix estimator
with a better {\em condition number} defined as the ratio between the largest and smallest eigenvalues \citep[see, e.g.,][]{Belsley:1980}.
In order to formally impose the boundedness condition \eqref{e:bounded1},
we modify the reference prior and use
\beqn
p_{\text{ref},\epsilon}(\Sigma)\propto\frac{1}{|\Sigma|\prod_{i<j}(r_i-r_j)}1_{A_\epsilon}(\Sigma),
\eeqn
where $A_\epsilon$ is the set of symmetric and positive definite matrices $\Sigma$
such that $\min_{i,j}|r_i-r_j|>\epsilon$. We take $\epsilon=10^{-6}$ in the implementation.
Then, the posterior distribution of $\Sigma$ is
\beqn
p_\epsilon (\Sigma|\v y)\propto\frac{\exp\Big\{-\frac12\tr\big(\Sigma^{-1}S\big)\Big\}}{|\Sigma|^{N/2+1}\prod_{i<j}(r_i-r_j)}1_{A_\epsilon}(\Sigma),
\eeqn
where $S=\sum_{i=1}^N(y_i-\bar{y})(y_i-\bar{y})'$ with $\bar{y}$ the sample mean.
To impose the boundedness condition, we select $g_0(\Sigma)$ to be the inverse Wishart distribution with scale matrix $S$ and degrees of freedom $N-p+1$,
\beqn
g_0(\Sigma)\propto\frac1{|\Sigma|^{N/2+1}}\exp\Big\{-\frac12\tr\big(\Sigma^{-1}S\big)\Big\}.
\eeqn
It is then straightforward to check that there exists a constant $\beta_0$ such that $g_0(\Sigma)\geq \beta_0p_\epsilon(\Sigma|\v y)$.

Because of the positive-definite constraint on $\Sigma$,
it is useful to transform from the space of positive-definite matrices
to an  unconstrained space using the one-to-one transformation
$\Sigma^*=\log(\Sigma)$ or $\Sigma=\exp(\Sigma^*)$,
where $\Sigma^*$ is a symmetric matrix \citep{Leonard:1992}.
We can generate $\Sigma$ by generating the unconstrained lower triangle of $\Sigma^*$.
Let $\Sigma^*=QR^*Q'$ where $R^*=\diag(r_1^*,...,r_p^*)$ with $r_1^*\geq r_2^*\geq...\geq r_p^*$
and $Q$ the orthogonal matrix.
Then $\Sigma = Q\diag(e^{r_1^*},...,e^{r_p^*})Q'$.
We are now working on an unconstrained space of $\Sigma^*$,
we therefore can fit multivariate mixtures of $t$ to the iterates.
The dimension of the parameter space is $d=p(p+1)/2=55$.
Note that, in order to  be able to generate $\Sigma^*$ from $g_0(\Sigma)$,
we first generate $\Sigma$ from $g_0(\Sigma)$ and then transform it to $\Sigma^*$ as follows.
Let $\Sigma=Q\Lambda Q'$ where $\Lambda=\diag(\lambda_1,...,\lambda_p)$ with $\lambda_1\geq...\geq\lambda_p>0$,
then $\Sigma^* = Q\diag(\log(\lambda_1),...,\log(\lambda_p))Q'$.

Each sampler was run for 500,000 iterations after discarding the first 500,000 burnin iterations.
To reduce the computing time for the ACMH sampler in such a long run, we stop updating the mixture of $t$ distribution after the burnin period.
Figure \ref{fig:covariance} plots the iterates from the three chains for the first and second marginals as well as the first 500 autocorrelations of
these iterates.
For the ARWMH sampler, mixing is very poor and the chain moves very slowly.
The AIMH and ACMH samplers mix better.
Table \ref{t:real table table} summarizes the results for the three samplers, and reports
the acceptance rates, the  IACT values and squared distances (both averaged over the marginals),
the maximum IACT (max IACT) and minimum squared distance (min Sq distance) (among the 55 marginal chains),
and the CPU times taken by each sampler. It also reports II/time = Number of iterations/ (IACT $\times $ CPU time),
min II/time = Number of iterations/ (max IACT $\times $ CPU time), Acc/IACT = $\times$ Acceptance rate / IACT
and min Acc/IACT = $\times$ Acceptance rate / (max IACT). In this example the ACMH sampler outperforms
the other two samplers on almost all the performance measures.

\begin{table}[h]
\centering
\vskip2mm
{\small
\begin{tabular}{lcccccc}
\hline\hline
&\multicolumn{3}{c}{Covariance Example}&\multicolumn{3}{c}{Spam Example}\\
			&ARWMH	&AIMH	&ACMH	&ARWMH	&AIMH	&ACMH\\
\hline
Acceptance rate (\%) 	&7.7	&27	&30	&20	&1.2	&14\\
Avg IACT		&476	&148	&28	&288	&333	&156\\
Max IACT		&574	&283	&42	&484	&411	&344\\
CPU Time (mins)		&21.9	&13.7	&36.5	&10.9	&5.9	&19.2\\
Avg Sq Dist ($\times10^4$)&0.2	&12	&17	&2121	&698	&7379\\
Min Sq Dist ($\times10^4$)&0.07	&8.6	&12	&1.5	&1.1	&13\\
Avg II/time 		&48	&246 	&489	&64	&102	&67 \\
Min II/time		&40	&129	&326	&40	&82	&30 \\
Avg Acc/IACT (times $10^3$)&16	&182	&1071	&69	&4	&90 \\
Min Acc/IACT (times $10^3$)&13	&95	&714	&41	&3	&41 \\
\hline\hline
\end{tabular}
}
\caption{Real data examples. Average II/time = Number of iterations/ (Average IACT $\times $ CPU time),
Min II/time = Number of iterations/ (Max IACT $\times $ CPU time),
Avg Acc/IACT = $1000 \times$ Acceptance rate / (Avg IACT), and
Min  Acc/IACT = $1000 \times$ Acceptance rate / (Max IACT).
}\label{t:real table table}
\end{table}

\begin{figure*}[h]
\centerline{\includegraphics[width=1\textwidth,height=.6\textwidth]{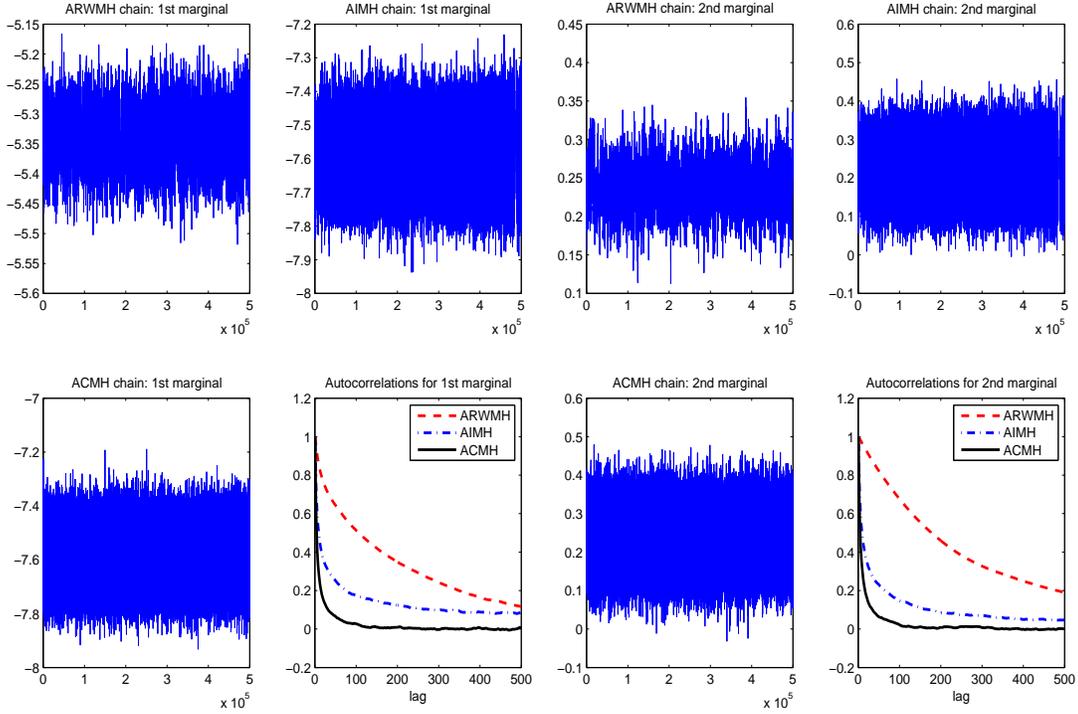}}
\caption{Covariance example: Plots of the iterates of the 1st and 2nd marginals and the first 500 autocorrelations of these iterates for the three adaptive samplers.
\label{fig:covariance}
}
\end{figure*}

%----------------------------------------%
\subsection{Spam filtering}
%----------------------------------------%
Automatic spam filtering is an important function for any email service provider.
The researchers at the Hewlett-Packard Labs created a spam email data set consisting of
4061 messages, each of which has been already been classified as an email or a spam
together with 57 attributes (predictors) which are relative frequencies of commonly occurring words.
The goal is to design a spam filter that can filter out spam before clogging the user's mailbox.
The data set is available at \texttt{http://www-stat.stanford.edu/$_{^\sim}$tibs/ElemStatLearn/}.

A suitable statistical model for this goal is logistic regression, where
the probability of being a spam, given the predictor vector $x$, is modeled as
\beqn
\mu(x,\theta) = P(y=1|x,\beta_0,\beta)=\frac{\exp(\beta_0+x'\beta)}{1+\exp(\beta_0+x'\beta)},
\eeqn
with $\theta=(\beta_0,\beta)'$ the coefficient vector.
For a future message with attributes $x$, using Bayesian inference, our goal is to estimate the posterior probability
that the message is classified as a spam,
\beq\label{e:mux}
\mu(x) = \E_{\t|D}(\mu(x,\theta))=\int \mu(x,\theta)p(\t|D)d\t,
\eeq
where $p(\t|D)$ denotes the posterior distribution of $\t$ given a training data set $D$.

We employ the weakly informative prior for $\theta$ proposed in \cite{Gelman:2008}.
The prior is constructed by first standardizing the predictors to have mean zero and standard deviation 0.5,
and then putting independent Cauchy distributions on the coefficients.
As a default choice, \cite{Gelman:2008} recommended a central Cauchy distribution with scale 10 for the intercept $\beta_0$
and central Cauchy distributions with scale 2.5 for the other coefficients.
This is a regularization prior and \cite{Gelman:2008} argue that it has many advantages;
in particular, it works automatically  without the need to elicit hyperparameters.
We set $g_0$ to be the  prior, which ensures that the boundedness condition \eqref{e:bounded1} is satisfied
because the maximum likelihood estimator for logistic regression exists.

We first use the whole data set and run each of the three samplers for 200,000 iterations after discarding 200,000 burnin iterations.
The dimension in this example is $d=58$.
Figure \ref{fig:spam} plots the iterates from the three chains for the first and second marginals as well as the first 500 autocorrelations of
these iterates.
As in the covariance estimation example,
Table \ref{t:real table table} summarizes the results for the three samplers.
The ACMH sampler outperforms the other two samplers except for the Avg II/time and Min II/time, where the results are mixed, because
of the longer running times. However, as noted at the end of Section~\ref{SS: comparison}, it is straightforward
to make the ACMH run faster.

\begin{figure*}[h]
\centerline{\includegraphics[width=1\textwidth,height=.6\textwidth]{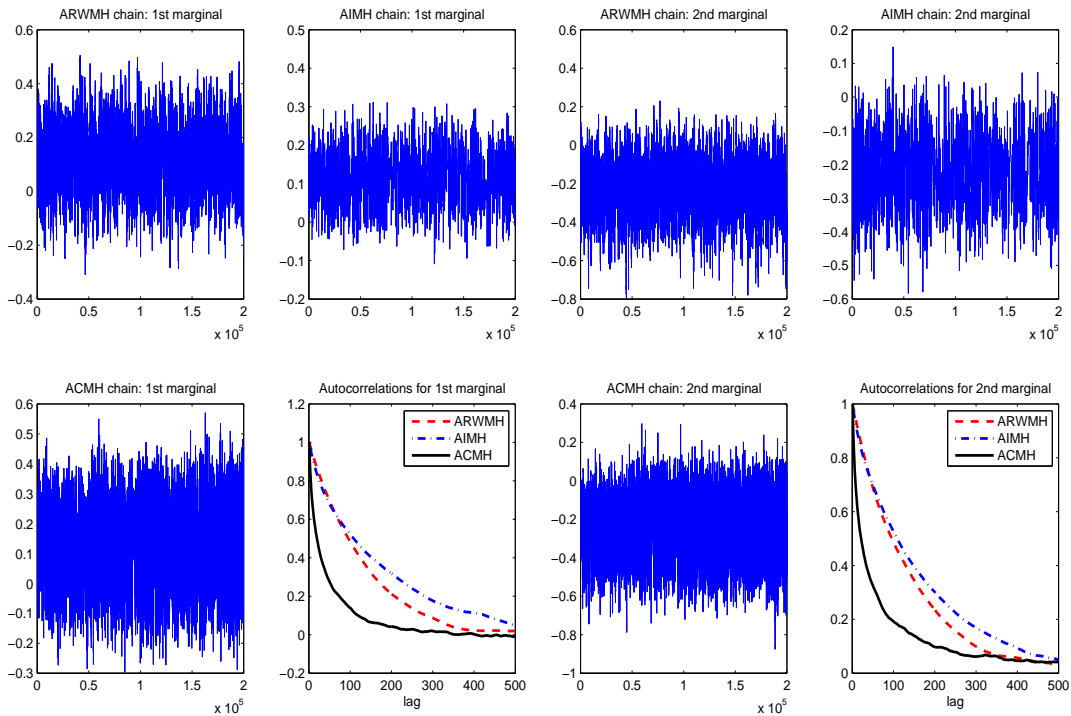}}
\caption{\label{fig:spam}
Spam email example: Plots of the iterations of the 1st and 2nd marginals and the first 500 autocorrelations of these iterations for the three adaptive samplers.
}
\end{figure*}

We also consider the predictive ability of the binary models estimated by the three chains.
The continuous ranked probability score (CRPS) is a widely used prediction measure in the forecasting community, see, e.g., \cite{Gneiting:2007} and \cite{Hersbach:2000}.
Let $F$ be the cumulative distribution function (cdf) of the predictive distribution in use
and $y$ be an actual observation. The CRPS is defined as
\beqn
\text{CRPS}(F|y) = \int_\mathbb{R}(F(z)-\I_{z\geq y})^2dz.
\eeqn
When $F_\mu$ is the cdf of a Bernoulli variable $Y$ with probability of success $\mu=P(Y=1)$,
the CRPS is given by
\beqn
\text{CRPS}(F_\mu|y=0) = \mu^2\;\;\text{and}\;\;\text{CRPS}(F_\mu|y=1) = (1-\mu)^2
\eeqn
with $\mu(x)$ given in \eqref{e:mux}.
Let $D^T$ be a test data set, we compute the CRPS (based on $D^T$) by
\beq\label{CRPSdef}
\text{CRPS}= \sum_{(x,y)\in D^T}\text{CRPS}(F_{\mu(x)},y).
\eeq
Under this formulation, it is understood that smaller CRPS means better predictive performance.

We now randomly partition the full data set into two roughly equal parts: one is used as the training set $D$ and the other as the test set $D^T$.
We would like to assess the performance of the samplers in terms of predictive accuracy.
To do so, we compute the CRPS for each sampler as in \eqref{CRPSdef}.
To take into account the randomness of the partition, we average the CRPS values over 5 such random partitions.
We first run each sampler for 50,000 iterations with another 50,000 burnin iterations.
The averaged CRPS values for the ARWMH, AIMH and ACMH are 135.23, 135.10, 133.14 respectively.
This result suggests that the ACMH has the best predictive accuracy for that number of iterations.
We then carried out longer runs for each sampler with 500,000 iterations after discarding 500,000 burnin iterations.
The averaged CRPS values for the ARWMH, AIMH and ACMH are now 123, 123.1, 122.8 respectively.
This means that all the samplers converge to the target if they are run long enough.

%======================================================================%
\section{Discussion}\label{Sec:conclusion}
%======================================================================%
This article develops a general-purpose adaptive correlated \MH{} sampler that will work well for multimodal as well as unimodal targets.
Its main features are the use of reversible proposals, the absence of a requirement for
diminishing adaptation, and having as its major component the mixture of $t$ model fitted by
Variational Approximation. The ACMH sampler combines exploratory and exploitative stages
and consists of various steps including correlated, random walk and Metropolis within Gibbs component-wise sampling
steps. This makes the sampler explore the target more effectively both globally and locally,
and improves the acceptance rate in high dimensional problems.
The convergence to the target is theoretically guaranteed without the need to impose
Diminishing Adaptation.

There are two important and immediate extensions of our work.
First, the ACMH sampler can be extended in a straightforward way to allow for a reversible proposal
whose invariant distribution is a mixture of $t$ that is constructed in a different way to the Variational
Approximation approach, e.g. as in \cite{Hoogerheide2012}. More generally, the ACMH sampler can be extended
to a reversible proposal whose invariant distribution is a more general mixture, e.g. a mixture of multivariate
betas or a mixture of Wishart densities, or a mixture of copula densities.
Second, the current article considers applications where the likelihood can
be evaluated explicitly (up to a normalizing constant). However, our methods apply equally well
to problems where the likelihood can only be estimated unbiasedly and adaptive MCMC is carried out
using the unbiased  estimate instead of the likelihood as in, for
example, \cite{andrieu:doucet:holenstein:2010} and \cite{Pitt:2012b}.

%======================================================================%
\section*{Acknowledgment}
%======================================================================%
The research of Minh-Ngoc Tran and Robert Kohn was partially supported
by Australian Research Council grant DP0667069. 
\bibliographystyle{apalike}
\bibliography{references}

%======================================================================%
\section*{Appendix}
%======================================================================%

%-----------------------------------------------------%
\begin{proof}[Proof of Theorem \ref{ergodic_theorem}]
%-----------------------------------------------------%
The proof is based on the split chain construction of \cite{Athreya1978}.
We can write
\beqn
p_{i}(x_{i-1},dx_{i})  = \beta\Pi(dx_{i})+(1-\beta\ )\nu_{i}(x_{i-1},dx_{i}),
\eeqn
where, by \eqref{assumption},
\beqn
\nu_i(x_{i-1},dx_i) = \frac{p_i(x_{i-1},dx_i)-\beta\Pi(dx_i)}{1-\beta}
\eeqn
is a transition distribution with invariant distribution $\Pi(\cdot) $.
By induction, it is easy to show that
\bqan
\P^{n}(x_{0},dx_{n}) &=&\int\P^{n-1}(x_0,dx_{n-1})p_n(x_{n-1},dx_n)\\
&=&(1-\beta\ )^{n}\nu^{n}(x_{0},dx_{n})+\left(1-(1-\beta\ )^{n}\right)  \Pi(dx_{n}),\;\;\text{for}\;\;n\geq1,
\eqan
where
\beqn
\nu^{n}(x_{0},dx_{n})  ={\int_{x_{1}}}\cdots
{\int_{x_{n-1}}}\nu_{1}(x_{0},dx_{1})\cdots\nu_{n}(x_{n-1},dx_{n}).
\eeqn
So
\beq\label{main_equality}
\P^{n}(x_{0},dx_{n})-\Pi(dx_{n})   =(1-\beta\ )^{n}(\nu^{n}(x_{0},dx_{n})-\Pi(dx_{n})).
\eeq
This implies that for any set $A\in\mathcal{E}$,
\beqn
|\P^{{n}}(x_{0},A)-\Pi(A)|=\left|\int_A(\P^{{n}}(x_{0},dx_n)-\Pi(dx_n))\right|\leq2(1-\beta)^n.
\eeqn
It follows that
\beqn
\|\P^{{n}}(x_{0},\cdot)-\Pi(\cdot)\|_{TV}\leq2(1-\beta)^{n},
\eeqn
for any initial $x_0$.
\end{proof}

We introduce some notation.
Denote by $\mu_0(\cdot)$ the distribution of the initial $x_0$.
For $j>i$, define,
\bqan
\P^{j|i}(x_{i},dx_{j})  &  =&\int_{x_{i+1}}\cdots\int_{x_{j-1}}p_{i+1}%
(x_{i},dx_{i+1})\cdots p_{j}(x_{j-1},dx_{j}).\\
\E_{\mu_{0}\P^{i}}(h)  &  =&\mu_{0}\P^{i}(h)=\int_{x_{0}}\int_{x_{i}}\mu_{0}%
(dx_{0})\P^{i}(x_{0},dx_{i})h(x_{i}).\\
\E_{\mu_{0}\P^{i}\P^{j|i}}(h\circ h)  &  =&\mu_{0}\P^{i}\P^{j|i}(h\circ h)=\int_{x_{0}}\int%
_{x_{i}}\int_{x_{j}}\mu_{0}(dx_{0})\P^{i}(x_{0},dx_{i})\P^{j|i}(x_{i}%
,dx_{j})h(x_{i})h(x_{j}).\\
\E_{\Pi}(h)  &  =&\Pi(h)=\int_{x}\Pi(dx)h(x).
\eqan
%-----------------------------------------------------%
\begin{lemma}\label{lemma_varSn}
%-----------------------------------------------------%
Suppose that the assumptions in Theorems~\ref{ergodic_theorem} and \ref{SLLN_theorem} hold. Then,
\begin{itemize}
\item[(i)]
\[
\frac{1}{n}\sum_{i=1}^{n}\E_{\mu_{0}\P^{i}}(h)=\E_{\Pi}(h)+O\left (\frac1n\right ).
\]
\item[(ii)]
\beqn
\frac{1}{n^{2}}{\sum_{i=1}^{n}}\E_{\mu_0\P^i}(h^2)= \E_\Pi(h^2) + O\left (\frac1n  \right ).
\eeqn
\item[(iii)]
\[
\frac{1}{n^{2}}\sum_{j=1}^{n}\sum_{i=1}^{n}\E_{\mu_{0}\P^{i}}(h)\E_{\mu_{0}\P^{j}%
}(h)=\E_{\Pi}(h)^2+  O\left (\frac1n  \right ).
\]
\item[(iv)]
\beqn
\frac{1}{n^{2}}{\sum_{j=2}^{n}}{\sum_{i=1}^{j-1}}  \E_{\mu_{0}\P^{i}\P^{j|i}}(h\circ h)= \E_\Pi(h)^2 +O\left (\frac1n  \right ).
\eeqn
\end{itemize}
\end{lemma}

\begin{proof}[Proof of Lemma \ref{lemma_varSn}]
Without loss of generality we take $ h \geq 0 $, because we can consider the positive and negative parts of $h$ separately. Let $\hmax$ be the maximum value of $h$.
To obtain Part~(i),
\begin{align*}
\left \vert \int \left ( P^i(x_0,dx) - \Pi(dx) \right ) h(x) \right \vert   & =
(1-\beta)^i \left \vert \int \left ( \nu^i(x_0,dx) - \Pi(dx) \right ) h(x)\right \vert
\leq 2 (1-\beta)^i \hmax \\
\frac{1}{n} \sum_{i=1}^n \left \vert \int \left ( P^i(x_0,dx) - \Pi(dx) \right ) h(x)
\right \vert  & \leq \frac{2\hmax }{n} \sum_{i=1}^n (1-\beta)^i
 =  O\left ( \frac{1}{n} \right ) .
 \end{align*}
Part (ii)  is obtained similarly.
To obtain Part (iii),
\begin{align*}
\mu_0P^i (h) \mu_0P^j(h)  - \Pi(h)^2 & = \left ( \mu_0P^i (h) - \Pi(h) \right ) \left ( \mu_0P^j (h) - \Pi(h) \right ) \\
& + \Pi(h) \left ( \mu_0P^j (h) - \Pi(h) \right )  + \left ( \mu_0P^i (h) - \Pi(h) \right ) \Pi(h)
\end{align*}
and the result follows from Part (i). Part (iv) is obtained similarly to Part (iii).
\end{proof}

%-----------------------------------------------------%
\begin{proof}[Proof of Theorem \ref{SLLN_theorem}]
%-----------------------------------------------------%
By (ii)-(iv) of Lemma \ref{lemma_varSn},
\bqan
\text{Var}\left(\frac{S_n}{n}\right)
%&=&\frac{1}{n^2}\sum_{i=1}^n\sum_{j=1}^n\left(\E\left(h(x_i)h(x_j)\right)-\E_{\mu_0\P^i}(h)\E_{\mu_0\P^j}(h)\right)\\
%&=&\frac{1}{n^2}\sum_{i=1}^n\sum_{j=1}^n\left(\E\left(h(x_i)h(x_j)\right)-\E_{\Pi}(h^2)\right)+\frac{1}{n^2}\sum_{i=1}^n\sum_{j=1}^n\left(\E_{\Pi}(h^2)-\E_{\mu_0\P%^i}(h)\E_{\mu_0\P^j}(h)\right)\\
&=&\frac{2}{n^{2}}{\sum_{j=2}^{n}}{\sum_{i=1}^{j-1}}\left(\E_{\mu_{0}\P^{i}\P^{j|i}}(h\circ h)-\E_\Pi(h)^2\right)+\frac{1}{n^{2}}{\sum_{i=1}^{n}}\left(\E_{\mu_0\P^i}(h^2)-\E_\Pi(h)^2\right)\\
&&\phantom{ccccccccccccc}+\frac{1}{n^2}\sum_{i=1}^n\sum_{j=1}^n\left(\E_{\Pi}(h)^2-\E_{\mu_0\P^i}(h)\E_{\mu_0\P^j}(h)\right)\\
&=&O\left (\frac1n\right ).
\eqan
The rest of the proof is similar to that in p. 326 of \cite{Grimmett:Stirzaker:2001}
\end{proof}

%-----------------------------------------------------%
\begin{proof}[Proof of Corollary \ref{corollary}]
%-----------------------------------------------------%
(i) Note that the acceptance probability at the $i$th iterate of the \MH{} algorithm is
\beqn
\alpha_i(z,x) = \left(1,\frac{\pi(z)q_i(x|z)}{\pi(x)q_i(z|x)}\right),
\eeqn
and the Markov transition distribution is
\beqn
p_{i}(x_{i-1},dx_{i}) = \alpha_{i}(x_{i},x_{i-1})q_{i}(x_i|x_{i-1})dx_{i}+\delta_{x_{i-1}}(dx_{i})\left(1-\int\alpha_{i}(z,x_{i-1})q_{i}(z|x_{i-1})dz\right)
\eeqn
From  $q_i(z|x)\geq\beta\pi(z)$ for all $x,z$, we can show that
\beqn
\alpha_i(z,x)q_i(z|x)\geq\beta\pi(z),
\eeqn
which implies that
\beqn
p_{i}(x_{i-1},dx_{i})\geq\beta\Pi(dx_i).
\eeqn
Therefore, the results in Theorems \ref{ergodic_theorem} and \ref{SLLN_theorem} follow.
Proof of (ii) is straightforward.
To prove (iii), note that
\beqn
p_{1,i}(x_{i-1},dx_{i})\geq\beta\Pi(dx_i),
\eeqn
therefore
\beqn
p_{i}(x_{i-1},dx_{i})\geq\omega\beta\Pi(dx_i),
\eeqn
which implies the results in Theorems \ref{ergodic_theorem} and \ref{SLLN_theorem}.
To prove (iv), note that
\beqn
p_{i}(x_{i-1},dx_{i})=p_{1,i}p_{2,i}(x_{i-1},dx_{i})\geq\int_{z}\beta\Pi(dz)p_{2,i}(z,dx_{i})=\beta\Pi(dx_{i}).
\eeqn
Also,
\beqn
p_{i}(x_{i-1},dx_{i})=p_{2,i}p_{1,i}(x_{i-1},dx_{i})\geq\int_{z}p_{2,i}(x_{i-1},dz)\beta\Pi(dx_{i}) = \beta\Pi(dx_{i}).
\eeqn
Theorems \ref{ergodic_theorem} and \ref{SLLN_theorem} then follow.
Part (v) is obtained similarly to Part (iv).
\end{proof}

%-----------------------------------------------------%
\begin{proof} [Proof of Lemma~\ref{lemma: properties of reversible densities}]
%-----------------------------------------------------%
\begin{itemize}\item [(i)]
$\zeta(x) T(z|x) = \zeta(x) \zeta(z) = \zeta(z)T(x|z)$.
\item [(ii)]
\begin{align*}
\zeta(x) T(z|x) = \zeta(x) \zeta(z_A, z_B)/\zeta(z_B) I(z_B = x_B) = I(z_B = x_B)\zeta(x)\zeta(z)/\zeta(z_B) =\zeta(z)T(x|z).
\end{align*}
\item [(iii)]
\begin{align*}
\zeta(x)T(z|x) = \int \zeta(x)T(z|x; \rho) \lambda(d \rho) = \int \zeta(z) T(x|z; \rho)  \lambda(d \rho)= \zeta(z) T(x|z).
\end{align*}
\end{itemize}
\end{proof}

%-----------------------------------------------------%
\begin{proof} [Proof of Lemma~\ref{lemma: reversible mixture}]
%-----------------------------------------------------%
\begin{itemize}
\item [(i)]
\begin{align*}
\zeta(x) T(z|x) & = \zeta(x)\sum_{k=1}%
^{G}\frac{\omega_k \zeta_k(x)}{\zeta(x)}  T_{k}(z|x) = \zeta(z) \sum_{k=1}%
^{G}\frac{\omega_k \zeta_k(z)}{\zeta(z)}  T_{k}(x|z) = \zeta(z)T(x|z).
\end{align*}
\item [(ii)]
$\omega(k|x) = \omega_k \zeta_k(x)/\zeta(x) =\omega_k$.
\item [(iii)]
This follows from part (ii) and Lemma~\ref{lemma: properties of reversible densities}(i).
\end{itemize}
\end{proof}

\begin{proof}[Proof of Corollary~\ref{corr: mixt cond}]
The proof follows from part (ii) of Lemma \ref{lemma: properties of reversible densities} and Lemma~\ref{lemma: reversible mixture}.
\end{proof}

\begin{proof} [Proof of Lemma~\ref{lem: reversible accept}]
The proof follows because  $T(z|x)/T(x|z) = \zeta(z)/\zeta(x)$.
\end{proof}

%-----------------------------------------------------%
\begin{proof} [Proof of Lemma~\ref{lem: reversible t}]
%-----------------------------------------------------%
Recall that $\phi_d(z; a, B) $ denotes a $d$-variate Gaussian density with mean vector $a$ and covariance matrix $B$ and let
$\IG(\lambda;\alpha,\beta)$ denote an inverse gamma density with shape parameter $\alpha$ and scale parameter $\beta$.
\begin{itemize}
\item[(i)] We first consider the case $\mu = 0$ and $\Sigma = I$.
Define the following densities $g(x|\lambda) = \phi_d(x; 0, \lambda I )$, $p(\lambda) = \IG(\lambda; \nu/2,\nu/2) $, and
$T(z|x; \rho, \lambda) = \phi_d(z;\rho x, \lambda(1-\rho^2)I) $.
Then, $p(\lambda|x) = \IG(\lambda; (\nu+d)/2, (\nu+x^\prime x)/2 ) $.
It is straightforward to establish that
\begin{align}
g(x|\lambda) p(\lambda) & = p(\lambda|x) \zeta(x;\psi),  \text{  where  } \zeta(x;\psi) = t_d(x; 0,I,\nu),    \label{eq: q revers}  \\
g(x|\lambda)T(z|x; \rho, \lambda)  & = g(z|\lambda) T(x|z; \rho, \lambda).   \label{eq: ar revers}
\end{align}
We define
\begin{align*}
T(z|x; \rho) & = \int T(z|x; \rho, \lambda) p(\lambda|x) d\lambda = t_d(z;\mu^\ast(x), \Sigma^\ast(x), \nu^\ast),
\end{align*}
where $$\mu^\ast (x) = \rho x, \text{  } \Sigma^\ast(x) = \frac{\nu + x^\prime x } { \nu +d} (1-\rho^2)I \text{ and }
\nu^\ast = \nu + d, $$ which is consistent with \eqref{eq: tilde defs} when $\mu = 0$ and $\Sigma = I$.
We now establish reversibility.
\begin{align*}
 \zeta(x;\psi)T(z|x; \rho) & =  \zeta(x;\psi) \int T(z|x; \rho, \lambda) p(\lambda|x) d\lambda \\
 & =  \zeta(x;\psi)\int \frac{ T(x|z; \rho, \lambda)g(z|\lambda) }{ g(x|\lambda) }  p(\lambda|x) d\lambda\quad\quad (\text{by \eqref{eq: ar revers}})\\
 & = \zeta(x;\psi)\int \frac{ T(x|z; \rho, \lambda)g(z|\lambda) }{ p(\lambda|x) \zeta(x;\psi) }  p(\lambda) p(\lambda|x) d\lambda\quad
(\text{by \eqref{eq: q revers}}) \\
& = \int   T(x|z; \rho, \lambda)g(z|\lambda) p(\lambda) d\lambda \\
& = \zeta(z;\psi) \int T(x|z; \rho, \lambda)p(\lambda|z) d\lambda\quad (\text{by \eqref{eq: q revers}})\\
& = \zeta(z;\psi)T(x|z; \rho)\; \text{   as required}.
 \end{align*}
The result for general mean $\mu$ and scale matrix $\Sigma$ is obtained by using the 
linear transformation for $\widetilde{x}  = \mu + \Sigma^\frac{1}{2} x$, yielding $\mutilde (\cdot) $
and $\Sigmatilde (\cdot)$ and noting that the additional Jacobian terms on the left and right sides are equal and so cancel out. 
\item[(ii)] The proof follows from Lemma~\ref{lemma: properties of reversible densities}(iii).
\end{itemize}
\end{proof}

\begin{proof} [Proof of Lemma \ref{eq: mixt t transition densities}]
(i) follows from Lemma~\ref{lemma: reversible mixture} (i); (ii) follows from Lemma~\ref{lem: reversible accept}.
\end{proof}

\begin{proof}[Proof of Lemma~\ref{lemma: rev cond densities}]
The proof follows from Lemma~\ref{lemma: properties of reversible densities}(ii) and Lemma \ref{lemma: reversible mixture}(i).
\end{proof}

\begin{proof} [Proof of Lemma~\ref{lemma: acmh proposal}]

(i) Follows from Part (ii) of Lemma~\ref{lemma: reversible mixture} because $g_M(z; \psihat(\calHn)) $
is the invariant density of both  $T_{g_M}^\text{CMH}(z|x; \psihat(\calHn) ) $
and $ T_{g_M}^\text{BS}(z|x; \psihat(\calHn) )$;
(ii) Follows from part (i) of Lemma~\ref{lemma: reversible mixture};
(iii) This follows from Part (iii) of Lemma~\ref{lemma: reversible mixture}.
(iv) This follows from Lemma~\ref{lem: reversible accept}.
\end{proof}

\end{document}